\documentclass{IEEEtran}
\IEEEoverridecommandlockouts
\usepackage[english]{babel}
\usepackage[utf8]{inputenc}
\usepackage[T1]{fontenc}

\usepackage{microtype}
\usepackage{setspace}
\usepackage{blindtext}
\usepackage{siunitx}
\usepackage{comment}
\usepackage{cite}
\usepackage{titlecaps}
\Addlcwords{or the with if a an of for and to -off off in not by via}
\usepackage{graphicx}
\usepackage{fancyhdr} 
\usepackage{color}
\usepackage{epsfig}
\graphicspath{{Figures/}}
\usepackage{tikz}
\usepackage{pgfplots,pgfplotstable}
\usetikzlibrary{patterns,fit,matrix}
\usepackage{tikzscale}
\usepackage{scalerel}
\usetikzlibrary{arrows,
	patterns,
	plotmarks,
	svg.path,
	shapes.multipart}
\usepgfplotslibrary{fillbetween}
\pgfplotsset{compat=newest,
	compat/bar nodes=1.8,
	every axis/.append style={
		label style={font=\Large},
		tick label style={font=\large} 
	}
}
\tikzstyle{int}=[draw, fill=black!10, minimum size=5em,thick]
\tikzstyle{init} = [pin edge={to-,thick,black}]
\usepackage{booktabs}

\usepackage{array}

\usepackage{enumitem}

\usepackage{amsthm,amsmath,amssymb,amsfonts}
\usepackage{relsize}
\usepackage{nicefrac}
\usepackage{bbm}
\usepackage{mathtools}
\usepackage[mathscr]{eucal}
\usepackage[short,c2]{optidef}
\usepackage{dsfont}

\usepackage{url}
\usepackage[colorlinks,citecolor=blue]{hyperref}
\usepackage[capitalize]{cleveref}

\allowdisplaybreaks
\newtheorem{corollary}{Corollary}
\newcommand{\orcid}[1]{\href{https://orcid.org/#1}{\includegraphics[scale=0.04]{orcid}}} % link to ORCID
\newcommand{\myParagraph}[1]{{\bf #1.}}

\newcommand{\journalVersion}[2]{#2} % use #1 for short, #2 for extended version

\newcommand*\diff{\mathop{}\!\mathrm{d}}

\newcommand{\ie}{\emph{i.e.,}\xspace}

\newcommand{\Real}[1]{ { {\mathbb R}^{#1} } }

\newcommand{\pr}[1]{\mathbb{P}\ls#1\rs}

\newcommand{\mean}[1]{\mathbb{E}\left[#1\right]}

\newcommand{\one}[1][]{\mathds{1}_{#1}}

\newcommand{\e}{\mathrm{e}}

\renewcommand{\maxof}[2]{#1 \vee #2}

\theoremstyle{plain}
\newtheorem{thm}{Theorem}
\newtheorem{lemma}{Lemma}

\theoremstyle{definition}
\newtheorem{definition}{Definition}

\newtheorem{ass}{Assumption}
\theoremstyle{remark}
\newtheorem{rem}{Remark}

\newcommand{\lr}{\left(}
\newcommand{\rr}{\right)}
\newcommand{\ls}{\left[}
\newcommand{\rs}{\right]}
\newcommand{\lb}{\left\lbrace}
\newcommand{\rb}{\right\rbrace}

\newcommand{\agents}{\mathcal{V}}
\newcommand{\leg}{\mathcal{L}}
\newcommand{\mal}{\mathcal{M}}

\newcommand{\state}[2]{x_{#1}^{#2}}
\newcommand{\stateleg}[1]{x_{#1}^{\leg}}
\newcommand{\statemal}[1]{x_{#1}^{\mal}}
\newcommand{\neigh}[2][\@empty]{\mathcal{N}%
	\ifx\@empty#1 {_{#2}} \else {_{#2}(#1)} \fi}
\newcommand{\neighaug}[2][\@empty]{\neigh[#1]{#2}\cup\{#2\}}
\newcommand{\sumneigh}[3][\@empty]{\sum_{#3\in\neigh[#1]{#2}}}
\newcommand{\sumneighaug}[3][\@empty]{\sum_{#3\in\neighaug[#1]{#2}}}
\newcommand{\lam}[1]{\lambda_{#1}}
\newcommand{\statecontrleg}[1]{\bar{x}_{#1}^{\leg}}
\newcommand{\statecontrmal}[1]{\bar{x}_{#1}^{\mal}}
\newcommand{\statelegtil}[1]{y_{#1}^{\leg}}

\newcommand{\statemaltil}[1]{y_{#1}^{\mal}}

\newcommand{\stateerr}[2]{\tilde{x}_{#2}^{#1}}
\newcommand{\statelegerr}[2]{\tilde{x}_{#2}^{#1,\leg}}
\newcommand{\statelegtruess}{x_{\text{ss}}^{\leg,*}}
\newcommand{\statemalerr}[2]{\tilde{x}_{#2}^{#1,\mal}}
\newcommand{\uleg}[1]{u_{\leg}\lr#1\rr}
\newcommand{\umal}[1]{u_{\mal}\lr#1\rr}

\newcommand{\Wleg}[1]{W_{#1}^{\leg}}
\newcommand{\Wmal}[1]{W_{#1}^{\mal}}
\newcommand{\Wlegtrue}{\overline{W}^{\leg}}
\newcommand{\Wlegerr}[2][\@empty]{\widetilde{W}%
	\ifx\@empty#1 {_{#2}^{\leg}} \else {_{#2,#1}^{\leg}} \fi}
\newcommand{\prodlam}[3]{\prod_{#1=#2}^{#3}(1-\lam{#1})}
\newcommand{\prodW}[3]{\prodlam{#1}{#2}{#3}\Wleg{#1}}
\newcommand{\prodWtrue}[3]{\prod_{#1=#2}^{#3}\Wlegtrue}
\newcommand{\prodWss}[3]{\prod_{#1=#2}^{#3}(1-\lam{#1})\Wlegtrue}
\newcommand{\prodlamW}[1]{\Pi_{#1}}
\newcommand{\prodinflam}[1]{\pi_{#1}}
\newcommand{\tf}{T_\text{f}}
\newcommand{\dmax}{d_{\text{M}}}
\newcommand{\vmin}{v_\text{m}}

\newcommand{\trust}[3]{\alpha_{#1#2}(#3)}
\newcommand{\trusthist}[3]{\beta_{#1#2}(#3)}
\newcommand{\meanleg}{E_\leg}
\newcommand{\meanmal}{E_\mal}

\addto\captionsenglish{\renewcommand{\figurename}{Fig.}}
\addto\captionsenglish{\renewcommand{\tablename}{TABLE}}

\newcommand{\canOmit}[1]{{\color{gray}#1}\xspace}
\newcommand{\MY}[1]{{\color{red}M: #1}} % Michal
\newcommand{\LB}[1]{{\color{blue}L: #1}} % Luca
\newcommand{\blue}[1]{{\color{blue}#1}}

\newcommand{\linkToPdf}[1]{\href{#1}{\blue{(pdf)}}}
\newcommand{\linkToPpt}[1]{\href{#1}{\blue{(ppt)}}}
\newcommand{\linkToCode}[1]{\href{#1}{\blue{(code)}}}
\newcommand{\linkToWeb}[1]{\href{#1}{\blue{(web)}}}
\newcommand{\linkToVideo}[1]{\href{#1}{\blue{(video)}}}
\newcommand{\linkToMedia}[1]{\href{#1}{\blue{(media)}}}
\newcommand{\award}[1]{\xspace} % omit awards
\addto\extrasenglish{}
\addto\extrasenglish{}
\addto\extrasenglish{}
\addto\extrasenglish{}

\title{
	\journalVersion{%
		The Role of Confidence for Trust-based Resilient Consensus
		}{%
		{\LARGE The Role of Confidence for Trust-based Resilient Consensus \\ (Extended Version)}
	}
}

\author{Luca~Ballotta\textonesuperior
	\xspace and~Michal~Yemini\texttwosuperior
	\thanks{This work has been partially supported
		by the Italian Ministry of Education, University and Research (MIUR) through the PRIN Project n. 2017NS9FEY
        and by the CARIPARO Foundation under Visiting Programme ``HiPeR''.
    }%
    \thanks{\textonesuperior Delft Center for Systems and Control, Delft University of Technology, 2628 CD Delft, The Netherlands
        \texttt{l.ballotta@tudelft.nl}.
    }%
    \thanks{\texttwosuperior Faculty of Engineering, Bar-Ilan University, Ramat-Gan 5290002 Israel
        \texttt{michal.yemini@biu.ac.il}.
    }
}

\begin{document}
	
	\maketitle
	
	%!TEX ROOT = ../resilient_consensus_trust.tex

\begin{abstract}
	\boldmath
	%\journalVersion{In this paper, we}{We}
    We consider a multi-agent system where agents aim to achieve a consensus despite interactions with malicious agents that communicate misleading information.
	Physical channels supporting communication in cyberphysical systems offer attractive opportunities to detect malicious agents,
	nevertheless,
	trustworthiness indications coming from the channel are subject to uncertainty and need to be treated with this in mind.
	We propose a resilient consensus protocol that incorporates trust observations from the channel and weighs them with a parameter that accounts for how confident an agent is regarding its understanding of the legitimacy of other agents in the network,
	with no need for the initial observation window $T_0$ that has been utilized in previous works.
	Analytical and numerical results show that
	\textit{(i)} our protocol achieves a resilient consensus in the presence of malicious agents and
	\textit{(ii)} the steady-state deviation from nominal consensus can be minimized by a suitable choice of the confidence parameter that depends on the statistics of trust observations.
	\journalVersion{}{This technical report contains proof details for the conference paper~\cite{acc}.
	}
\end{abstract}
	%!TEX ROOT = ../resilient_consensus_trust.tex

\section{Introduction}\label{sec:intro}

Consensus in multi-agent systems is an essential tool in many applications, including distributed control and multi-robot coordination.
However,
the consensus protocol is fragile to outliers and easily fails in the presence of agents that do not behave according to it --- for example in the adversarial case.

To tame \textit{malicious agents} and recover a \textit{resilient consensus} among \textit{legitimate agents},
several strategies have been proposed in the literature.
One common method to achieve this goal is the Weighted-Mean Subsequence Reduced (W-MSR) algorithm~\cite{6481629},
which has been adapted to many application domains~\cite{Usevitch20tac-resilientLeaderFollower,Shang20automatica-resilientConsensus}.
Other strategies that have been recently proposed use different rules to filter out suspicious data,
such as the similarity between two agents' states~\cite{Baras19med-trust},
or leverage enhanced network structure,
such as secured agents~\cite{Abbas18tcns-trustedNodes}. % that cannot be hacked.

Recovering a resilient consensus purely based on the data exchanged among agents is in general a challenging task.
A notable limitation on the theoretical guarantees of W-MSR is that the communication graph needs to enjoy a connectivity property,
called $r$-robustness,
that ensures a pervasive information flow among legitimate agents.
%Unfortunately, a high enough $r$-robustness may require dense networks, and it cannot be verified in polynomial time w.r.t. the number of agents~\cite{7447011,9993257}.
Unfortunately,
ensuring a sufficiently high $r$-robustness may require dense network topologies, which cannot be verified in polynomial time with respect to (w.r.t.) the number of agents~\cite{7447011,9993257}.
Thus,
in real-world applications and especially in large networks,
W-MSR may not lead to a consensus.

In contrast to data-centered approaches,
recent works~\cite{ICRA2019_switching,Yemini22tro-resilienceConsensusTrust,9481121,CDC_dist_optimization,journal_distributed_optimization,hadjicostis2022trustworthy} have proposed to use \textit{physical} information of transmissions to boost resilience in distributed cyberphysical systems,
leveraging the fact that this source of information is independent from the exchanged data.
Cyberphysical systems are widely adopted in applications,
from robot teams to smart grids.
In such systems,
communication occurs over physical channels that can be used to extract information used to assess the validity of a transmission:
for instance,
wireless signals can be analyzed to detect manipulated messages~\cite{gil2017guaranteeing,8758381,gil2023physicality}. %received data is consistent with the signal itself,

However, while using physical transmission channels as a source of information for legitimacy of received messages allows one to decouple the consensus task from the detection of potential adversaries within the network,
this information is usually uncertain~\cite{gil2017guaranteeing},
partially hindering its usefulness if this is not properly accounted for.
This calls for attention in embedding the physical trustworthiness indications into the design of a resilient consensus protocol.

In this paper,
we draw inspiration from the trust-based protocol in~\cite{Yemini22tro-resilienceConsensusTrust} and the competition-based approach in~\cite{ballottaExtended},
and propose a novel algorithm that integrates the notion of \textit{trust},
coming from the physical channel,
and the concept of \textit{confidence},
which counterbalances the uncertainty in agent classification.
This integration allows us to circumvent two limitations of the previous algorithms: % studied in~\cite{Yemini22tro-resilienceConsensusTrust,ballottaCDC22compColl}:
on the one hand,
we do not need a time window $T_0>0$ of trust observations %wait before running the protocol,
as in~\cite{Yemini22tro-resilienceConsensusTrust};
on the other hand,
the agents achieve an asymptotic consensus,
differently from the data-driven context in~\cite{ballottaExtended}.
Specifically,
the proposed protocol anchors the agents to their initial condition through a time-varying weight $\lam{t}$ that reflects how confident an agent is about the trustworthiness of its neighbors:
owing to the competition-based approach,
this strategy avoids the agents to be misled through misclassification of neighbors and enhances resilience in the face of both unknown malicious agents and uncertain information from the physical channel.
Moreover,
we show that the confidence parameter can be tuned to optimize performance:
analytical and numerical results indicate that $\lam{t}$ should decay according to the average time the agents need to correctly classify their neighbors.

The rest of this paper is organized as follows.
\autoref{sec:system-model} presents the system model and the problem formulation,
while~\autoref{sec:resilient-protocol} introduces the proposed resilient consensus protocol and mathematical models for trust and confidence.
Then,
\autoref{sec:performance-analysis} provides theoretical guarantees offered by the protocol,
focusing on convergence (\autoref{sec:convergence}) and asymptotic deviation from the nominal consensus (\autoref{sec:deviation}).
Finally,
\autoref{sec:simulations} presents numerical simulation results that corroborate the analysis and prove our protocol effective.
	%!TEX ROOT = ../resilient_consensus_trust.tex

\section{Setup}
			%!TEX ROOT = ../resilient_consensus_trust.tex

\subsection{System Model and Problem Formulation}\label{sec:system-model}

\myParagraph{Network}
We consider a multi-agent system composed of $N$ agents equipped with scalar-valued states:
we denote the state of agent $i$ at time $t$ by $\state{t}{i}\in\Real{}$,
with $i \in \agents \doteq \{1,\dots,N\}$,
and the vector with all stacked states by $\state{t}{}\in\Real{N}$.
The agents can communicate and exchange their states through a fixed communication network,
modeled as a graph $\mathcal{G} = (\agents, \mathcal{E})$.
Each element $e = (i,j) \in \mathcal{E}$ indicates communication edge between agents $i$ and $j$:
if $(i,j)\in\mathcal{E}$,
it means that agent $j$ can transmit data to agent $i$ through a direct link.

In the network,
$L$ agents truthfully follow a designated protocol (\textit{legitimate agents} $ \leg\subset\agents $)
while $M=N-L$ agents behave arbitrarily (\textit{malicious agents} $ \mal\subset\agents $), % and may deviate from the protocol,
potentially disrupting the task executed by legitimate agents.
%Without loss of generality,
We set the labels of legitimate and malicious agents as $\leg = \{1,\dots,L\}$ and $\mal=\{L+1,\dots,N\}$
and denote their collective states respectively by $\stateleg{t}\in\Real{L}$ and $\statemal{t}\in\Real{M}$.
We denote by $\dmax$ the maximal (in-)degree of legitimate agents,
with $\dmax<N$.

We assume that the states $\state{t}{i}$ are bounded for every $i\in\agents$ and $t\geq0$.
In the derivation,
we use the following assumption.
\begin{ass}[State bound]\label{ass:initial-state-bound}
	It holds $\max_{i\in\agents,t\ge0}|\state{i}{t}|\le\eta$.
\end{ass}

\myParagraph{Consensus Task}
The legitimate agents aim to achieve a consensus.
The nominal consensus value is determined by their initial states $\stateleg{0}$ and by the ideal communication network without malicious agents.
Specifically,
let $\neigh{i}\in\agents$ denote the neighbors of agent $i$ in the communication network $\mathcal{G}$,
\ie $\neigh{i} \doteq \{j\in\agents : (i,j)\in\mathcal{E}\}$,
and consider the nominal matrix $\Wlegtrue\in\Real{L\times L}$ with weights defined as follows for $i,j\in\leg$:
\begin{equation}\label{eq:weights-true}
	\ls\Wlegtrue\rs_{i j}=
	\begin{cases}
		\frac{1}{\left|\neigh{i}\cap\leg\right|+1} & \text { if } j \in \neigh{i}\cap\leg, \\ 
		0 & \text { if } j \notin \neighaug{i},\\
		1-\sumneigh{j}{i} \ls\Wlegtrue\rs_{ij} & \text { if } j=i.
	\end{cases}
\end{equation}
Ideally,
the legitimate agents should disregard messages sent by malicious agents
(\ie set their weights to zero)
and run the following nominal consensus protocol starting from $\stateleg{0}$:
\begin{equation}\label{eq:consensus-protocol-nominal}\tag{NOM}
	\stateleg{t+1} = \Wlegtrue\stateleg{t}, \qquad t\ge0.
\end{equation}
Unfortunately,
the identity of malicious agents is unknown to legitimate agents,
so that these cannot implement the weights~\eqref{eq:weights-true} and the protocol~\eqref{eq:consensus-protocol-nominal}.
%In this way,
%the nominal consensus value is given the Perron-Frobenius eigenvector of $\Wlegtrue$:
In the next section,
we propose a \textit{resilient consensus} protocol aimed at recovering the final outcome of~\eqref{eq:consensus-protocol-nominal} in the face of malicious agents.
			%!TEX ROOT = ../resilient_consensus_trust.tex

\subsection{Resilient Consensus Protocol}\label{sec:resilient-protocol}

In this work,
we propose the following resilient protocol to be implemented by each legitimate agent $i\in\leg$ for $t\ge0$:
\begin{equation}\label{eq:update-rule-regular}\tag{RES}
	\state{t+1}{i} = \lam{t}\state{0}{i} + \lr1-\lam{t}\rr\sumneighaug{i}{j}w_{ij}(t)\state{t}{j}.
\end{equation}
Rule~\eqref{eq:update-rule-regular} uses two key ingredients.
The weights $ w_{ij}(t) \in[0,1]$ are computed online based on \textit{trust information} that agent $i$ collects about its neighbor $j$ overtime.
The time-varying parameter $ \lam{t} \in\ [0,1] $ accounts for how \textit{confident} agent $i$ feels about the trustworthiness of its neighbors.
In the following,
we describe these two features in detail.
We note that the parameter $\lam{t}$ is new w.r.t. to previous work~\cite{Yemini22tro-resilienceConsensusTrust} and a major objective in this work is to analytically characterize the impact of this ``confidence'' term on mitigating the effect of malicious agents when consensus protocol~\eqref{eq:update-rule-regular} starts from time $0$ (\ie no observation window as in~\cite{Yemini22tro-resilienceConsensusTrust} is present).

\myParagraph{Trust}
We are interested in the case where each transmission from agent $j$ to agent $i$ can be tagged with an observation $\trust{i}{j}{t}\in[0,1]$ of a random variable $\alpha_{ij}$.
%This random variable represents a \emph{probability of trust} that agent $i$ can give to its neighbor $j\in \neigh{i}$. 

\begin{definition}[Trust variable $\alpha_{ij}$]\label{def:alpha}
	For every $i\in\leg$ and $j\in\neigh{i}$,
	the random variable $\alpha_{ij}$ taking values in the interval $[0,1]$ represents the probability that 
	agent $j\in\neigh{i}$ is a trustworthy neighbor of agent $i$.
	We denote the expected value of $\alpha_{ij}$ by $\meanleg\doteq\mean{\alpha_{ij}} - \nicefrac{1}{2}$ for legitimate transmissions
	and by $\meanmal\doteq\mean{\alpha_{ij}} - \nicefrac{1}{2}$ for malicious ones.
	We assume the availability of observations $\trust{i}{j}{t}$ of $\alpha_{ij}$ through $t\ge0$.
\end{definition}

We refer to~\cite{AURO} for a concrete example of such an $\alpha_{ij}$ variable. 
Intuitively, a random realization $\trust{i}{j}{t}$ % of $\alpha_{ij}$ 
contains useful trust information if the legitimacy of the transmission can be thresholded.  
We assume that a value $\trust{i}{j}{t}>\nicefrac{1}{2}$ indicates a legitimate transmission and $\trust{i}{j}{t}<\nicefrac{1}{2}$ a malicious transmission in a stochastic sense (miscommunications are possible).  
The value $\trust{i}{j}{t}=\nicefrac{1}{2}$ means that the observation is completely ambiguous and contains no useful trust information for the transmission at time $t$.

\myParagraph{Weights}
The weights $w_{ij}(t)$ in~\eqref{eq:update-rule-regular} are chosen according to the history of trust scores $ \trust{i}{j}{t} $.
By defining the aggregate trust of communications from agent $ j $ to agent $ i $ as
\begin{equation}\label{eq:trust-history}
	\trusthist{i}{j}{t}=\sum_{s=0}^t\lr\trust{i}{j}{s} - \dfrac{1}{2}\rr, \quad i \in \mathcal{L}, j \in \neigh{i},
\end{equation}
we define the \textit{trusted neighborhood} of agent $i$ at time $ t $ as
\begin{equation}\label{eq:trusted-neighborhood}
	\neigh[t]{i}\doteq\lb j \in \neigh{i}:\trusthist{i}{j}{t} \geq 0 \rb.
\end{equation}
Then,
%for $i\in\leg$,
the weights in~\eqref{eq:update-rule-regular} are assigned online as follows:
\begin{equation}\label{eq:weights-trust}
	w_{i j}(t)= 
	\begin{cases}
		\frac{1}{\left|\neigh[t]{i}\right|+1} & \text { if } j \in \neigh[t]{i}, \\ 
		0 & \text { if } j \notin \neighaug[t]{i}, \\ 
		1-\sumneigh{j}{i} w_{ij}(t) & \text { if } j=i.
	\end{cases}
\end{equation}
The weighing rule above attempts to recover the nominal weights~\eqref{eq:weights-true} as time proceeds.
In particular,
the trusted neighborhood $\neigh[t]{i}$ is designed to reconstruct the set $\neigh{i}\cap\leg$ leveraging trust information collected by agent $i$ overtime.

\myParagraph{Confidence}
Because trust observations $\trust{i}{j}{t}$ may misclassify transmissions,
the weights computed as per~\eqref{eq:weights-trust} may not immediately recover the true weights:
in fact,
even assuming that a sufficient number of transmissions can give a clear indication about the trustworthiness of a neighbor,
a legitimate agent needs to act cautiously as long as it is unsure about the trust information collected in order to not be misled by erroneous classifications.
To this aim,
we modify the standard consensus rule by adding the parameter $\lam{t}$ in~\eqref{eq:update-rule-regular}
%with the goal of tuning the weights computed so far to counterbalance the chance of misclassifications.
that anchors the legitimate agents to their initial condition and refrains them from fully relying on the neighbors' states.

Intuitively,
agent $i$ accrues knowledge about the trustworthiness of its neighbors as more trust-tagged transmissions have been received.
This intuition can in fact be formalized by upper bounding the probability of misclassifying a neighbor.
\begin{ass}[Trust observations are informative]\label{ass:trust-meaningful}
	Legitimate (malicious) transmissions are classified as legitimate (malicious) on average.
	Formally,
	$\meanleg>0$ and $\meanmal<0$.
\end{ass}
\begin{lemma}[Decaying misclassification probability~\cite{Yemini22tro-resilienceConsensusTrust}]\label{lem:misclassification-probability}
	\begin{equation}\label{eq:prob-misclassification}
		\begin{gathered}
			\pr{\trusthist{i}{j}{t} < 0} \le \e^{-2\meanleg^2(t+1)} \quad \forall i\in\leg, j\in\neigh{i}\cap\leg\\
			\pr{\trusthist{i}{j}{t} \ge 0} \le \e^{-2\meanmal^2(t+1)} \quad \forall i\in\leg, j\in\neigh{i}\cap\mal.
		\end{gathered}
	\end{equation}
\end{lemma}
\cref{lem:misclassification-probability} implies that,
under~\cref{ass:trust-meaningful} that trust values $\trust{i}{j}{t}$ are informative, %actually useful,
the legitimate agents can infer which neighbors are trustworthy with higher confidence guarantees overtime.
On the other hand,
the early iterations of the protocol have higher chance of misclassifications.
To counterbalance this fact and make updates resilient,
we design the parameter $ \lam{t} $ as \textit{decreasing} with time.
This way,
early updates are conservative and not much sensitive to misclassifications 
($ \lam{t} \lesssim 1 $ for small $ t $),
while late updates rely almost totally on the neighbors confidently classified as legitimate
($ \lam{t} \gtrsim 0 $ for large $ t $).

\myParagraph{Discussion - Trust and Confidence}
The update rule~\eqref{eq:update-rule-regular} leverages the two fundamental concepts of \emph{trust} and \emph{confidence},
which are used together in an intertwined manner.

The works~\cite{Yemini22tro-resilienceConsensusTrust,CDC_dist_optimization,journal_distributed_optimization}
show how to utilize physics-based trust observations to help a legitimate agent decide which neighbors it should rely on as it runs the protocol.
Nonetheless, 
at each step,
the agent can either trust a neighbor or not and it does not scale the weights given to trusted neighbors relatively by how confident it is on the decision. 
Furthermore, 
in the work \cite{Yemini22tro-resilienceConsensusTrust} the deviation from the nominal consensus value is strongly tied to an initial observation window $T_0$ %of trust values 
where the agents do not trust any of their neighbors and only collect trust observations to  choose wisely what neighbors to trust in the first data update round.
This  length $T_0$ value is not straightforward to choose when the number of overall rounds varies and is not guaranteed in advance.
In contrast,
this work introduces the parameter $\lam{t}$ to capture the confidence that an agent has about the legitimacy of its neighbors, 
propose a softer approach to the clear-cut observation window used in~\cite{Yemini22tro-resilienceConsensusTrust} where agents do not trust one another,
and explores the role of such a confidence parameter to opportunistically tune the weights assigned to the neighbors.
In particular,
the formulation~\eqref{eq:update-rule-regular} highlights that the agent tunes the weights given to trusted neighbors scaling them by $(1-\lam{t})$.

% confidence in the classification of their trusted neighbors.

The use of $\lam{t}$ draws inspiration from previous work~\cite{ballottaCDC22compColl,ballottaExtended}
where the Friedkin-Johnsen model~\cite{FJdynamics} is used to achieve resilient average consensus,
intended as the minimization of the mean square deviation.
Contrarily to the trust-based works mentioned above,
the latter references do not use information derived from physical transmissions but study a robust update rule within a data-based context.
%In particular,
The updates in~\cite{ballottaCDC22compColl,ballottaExtended} use a constant parameter $\lambda$ (interpreted as \emph{competition} among agents) that mitigates the influence of malicious agents
by forcefully anchoring the legitimate agents to the initial condition,
ruling out the possibility of getting arbitrarily close to the nominal consensus.
In this work,
we use a source of information independent of the data (because it derives from physical transmissions)
to make the competition-based rule more flexible and able to recover a consensus. % at steady state.
	%!TEX ROOT = ../resilient_consensus_trust.tex

\section{Performance Analysis}\label{sec:performance-analysis}

Let $W_t\in\Real{L\times N}$ denote the matrix with weights~\eqref{eq:weights-trust},
\ie $[W_t]_{ij} = w_{ij}(t)$,
and consider the following partition:
\begin{equation}
	W_t = \ls\begin{array}{c|c}
		\Wleg{t} &\Wmal{t}
	\end{array}\rs, \quad \Wleg{t}\in\Real{L\times L}, \quad \Wmal{t}\in\Real{L\times M}.
\end{equation}
The protocol~\eqref{eq:update-rule-regular} can be rewritten as follows:
\begin{equation}\label{eq:dynamics}
	\begin{aligned}
		\stateleg{t+1}	&= \lam{t}\stateleg{0} + (1-\lam{t})\begin{bmatrix}
			\Wleg{t} & \Wmal{t}
		\end{bmatrix}\begin{bmatrix}
			\stateleg{t}\\
			\statemal{t}
		\end{bmatrix}\\
						&= \statecontrleg{t} + \statecontrmal{t}
	\end{aligned}
\end{equation}
where we define the state contributions due to legitimate and malicious agents' \textit{inputs},
respectively as
\begin{subequations}
	\begin{align}
		&\statecontrleg{t} \doteq \prodW{k}{0}{t}\stateleg{0} + \sum_{k=0}^t \lr\prodW{s}{k+1}{t}\rr\lam{k}\stateleg{0} \label{eq:leg-state-contr-leg}\\
		&\statecontrmal{t} \doteq \sum_{k=0}^t \lr\prodW{s}{k+1}{t}\rr (1-\lam{k})\Wmal{k}\statemal{k}.\label{eq:leg-state-contr-mal}
	\end{align}
\end{subequations}
In the following,
%for the sake of analysis,
we assume the parameter $\lam{t}$ has expression
\begin{equation}\label{eq:lambda}
	\lam{t} = c\e^{-\gamma t}, \qquad 0 < c < 1, \ \gamma > 0.
\end{equation}
%Note that,
We set $\lam{t}$ decreasing overtime to enable resilient updates at the beginning.
We choose~\eqref{eq:lambda} mainly to make analysis tractable.
We are mostly concerned with how the coefficient $\gamma$,
which dictates how fast $\lam{t}$ decays to zero,
affects the deviation from~\eqref{eq:consensus-protocol-nominal}.
Nonetheless,
we argue that the insights offered by our analysis apply to other choices of $\lam{t}$.
Also,
the misclassification probabilities~\eqref{eq:prob-misclassification} decay exponentially,
suggesting that~\eqref{eq:lambda} could be a good match with the trust statistics.

\journalVersion{Due to space limitations,
in the following we report only the main steps.
The detailed analysis is available at~\cite{arxiv}.}{}
		%!TEX ROOT = ../resilient_consensus_trust.tex

\subsection{Convergence to Consensus}\label{sec:convergence}

\begin{corollary}[\protect{\!\!\cite[Proposition~1]{Yemini22tro-resilienceConsensusTrust}}]\label{cor:finite_Tf_exsistence_as}
\cref{lem:misclassification-probability} implies that there exists almost surely (a.s.) a random finite time $\tf\ge0$ such that the estimated weights $\Wleg{t}$ equal the true weights $\Wlegtrue$ for all $t\ge\tf$.
\end{corollary}
Also,
under the mild assumption that the subgraph induced by the legitimate agents is connected,
the following fact holds.
\begin{lemma}[\protect{\!\!\cite[Lemma~1]{Yemini22tro-resilienceConsensusTrust}}]\label{lem:W-primitive}
	The matrix $\Wlegtrue$ is primitive and there exists a stochastic vector $v$ such that $\lr\Wlegtrue\rr^\infty=\one v^\top$.
\end{lemma}

Let $\maxof{a}{b} \doteq \max\{a,b\}$.
For every finite $k_0\ge0$,
it almost surely holds that
\begin{equation}\label{eq:prod-leg-influence-infty-1}
	\prodW{k}{k_0}{\infty} 	= \one v^\top \prodinflam{k_0} \prodlamW{k_0}
\end{equation}
where $v$ is the Perron eigenvector of $\Wlegtrue$ (see~\cref{lem:W-primitive}) and %$\lr\Wlegtrue\rr^{\infty}=\one v^\top$ and
\begin{equation}\label{eq:def-prod-infty-lam}
	\prodinflam{k_0} \doteq \prodlam{k}{\maxof{k_0}{\tf}}{\infty}, \quad \prodlamW{k_0} \doteq \prodW{k}{k_0}{(\maxof{k_0}{\tf})-1}.
\end{equation}

If $\lam{t}\neq1 \;\forall t$,
it holds $\prodinflam{k_0} > 0$ if and only if $\sum_{k=\maxof{k_0}{\tf}}^{\infty}\lam{k}$ converges~\cite{infProductConvergence},
which is the case under~\eqref{eq:lambda} for every $\gamma>0$.

\myParagraph{Contribution by Legitimate Agents}
From the definition~\eqref{eq:leg-state-contr-leg}, ~\eqref{eq:prod-leg-influence-infty-1}, and \cref{cor:finite_Tf_exsistence_as} it holds almost surely  at the limit that
\begin{equation}\label{eq:leg-agents-influence-convergence}
	\begin{aligned}
		\statecontrleg{\infty}	\journalVersion{}{&= \prodW{k}{0}{\infty}\stateleg{0} + \sum_{k=0}^\infty\lr\prodW{s}{k+1}{\infty} \rr\lam{k} \stateleg{0}\\}
								&= \one v^\top \underbrace{\lr\prodinflam{0}\prodlamW{0} \stateleg{0} + \sum_{k=0}^\infty \prodinflam{k+1}\prodlamW{k+1}\lam{k}\stateleg{0}\rr}_{\doteq\statelegtil{}}.
	\end{aligned}
\end{equation}
In view of~\eqref{eq:lambda},
the coordinates of $\statelegtil{}$ are finite because so are the coordinates of $\stateleg{0}$ and
the matrices $\prodlamW{k}$ are sub-stochastic.

\myParagraph{Contribution by Malicious Agents}
From the definition~\eqref{eq:leg-state-contr-mal},~\eqref{eq:prod-leg-influence-infty-1},
and \cref{cor:finite_Tf_exsistence_as} it holds almost surely  at the limit that
\begin{equation}\label{eq:mal-agents-influence-convergence}
	\begin{aligned}
		\statecontrmal{\infty}	\journalVersion{}{&= \sum_{k=0}^\infty \lr\prodW{s}{k+1}{\infty}\rr (1-\lam{k})\Wmal{k}\statemal{k}\\}
								&= \one v^\top\underbrace{\sum_{k=0}^{\tf-1} \prodinflam{k+1} \prodlamW{k+1} (1-\lam{k})\Wmal{k}\statemal{k}}_{\doteq\statemaltil{}}
	\end{aligned}
\end{equation}
where $\statemaltil{}$ almost surely sums a finite number of vectors.

Combining~\eqref{eq:dynamics} with~\eqref{eq:leg-agents-influence-convergence}--\eqref{eq:mal-agents-influence-convergence},
we conclude that legitimate agents a.s. converge to the consensus $\stateleg{\infty} = \one v^\top(\statelegtil{}+\statemaltil{})$ where it can be shown that,
by \cref{ass:initial-state-bound}, $\|\statelegtil{}\|<\infty$ and $\|\statemaltil{}\|<\infty$ for any choice of $\lambda_k\in[0,1],\: \forall k\geq0$,
and $\statelegtil{}$ and $\statemaltil{}$ are nonzero if the sequence $\{\lam{k}\}_{k\ge0}$ is summable.
		%!TEX ROOT = ../resilient_consensus_trust.tex

\subsection{Deviation from Nominal Consensus}\label{sec:deviation}

After assessing that the legitimate agents asymptotically achieve a consensus with probability $1$ (w.p.$1$),
we wish to evaluate the steady-state deviation from the nominal consensus value,
which is the one induced by the nominal weight matrix $\Wlegtrue$.
We quantify the deviation of agent $i\in\leg$ at time $t$ as follows:
\begin{equation}\label{eq:deviation}
	\stateerr{i}{t} \doteq \left|\state{t}{i} - \statelegtruess\right| = \left|\ls\stateleg{t} - \one\statelegtruess\rs_i\right|
\end{equation}
where $\statelegtruess \doteq v^\top\stateleg{0}$ is the nominal consensus value of legitimate agents at steady state.
In particular,
we are interested in upper-bounding the probability of the event that the deviation of legitimate agent $i$ from the nominal consensus value is greater than a threshold $\epsilon$,
\ie
\begin{equation}\label{eq:max-ss-deviation}
	\limsup_{t\rightarrow\infty}\stateerr{i}{t} > \epsilon.
\end{equation}
To this end,
in \autoref{sec:deviation-leg} and \autoref{sec:deviation-malicious} we respectively evaluate the state contributions of legitimate and malicious agents,
and then combine their deviations to bound the probability of~\eqref{eq:max-ss-deviation}.

\begin{rem}
	By virtue of the consensus reached at steady state by legitimate agents according to~\autoref{sec:convergence},
	all state trajectories have a well-defined limit (\ie the consensus value) almost surely, that is equal for all legitimate agents.
	This means that,
	in practice (w.p. $1$),
	$\limsup_{t\rightarrow\infty}\stateerr{i}{t} = \lim_{t\rightarrow\infty}\stateerr{i}{t}$.
\end{rem}

Evaluating the deviation from nominal is helpful to achieve analytical intuition that can help to design the parameter $\lam{t}$.
Intuitively,
small values of $\gamma$ in~\eqref{eq:lambda} refrain the legitimate agents from collaborating with trusted neighbors for longer time,
which should help when the trust scores $\trust{i}{j}{t}$ are rather uncertain,
while large values of $\gamma$ turn~\eqref{eq:update-rule-regular} into the standard consensus protocol after a few iterations,
and should suit cases when the true weights are quickly recovered.
				%!TEX ROOT = ../resilient_consensus_trust.tex

\subsubsection{Legitimate Agents}\label{sec:deviation-leg}
In view of the setup in~\autoref{sec:resilient-protocol},
the only correct contribution to the state of any legitimate agent is the information coming from other legitimate agents,
which ideally leads to the true consensus value $\statelegtruess$.
Hence,
we define the deviation term due to legitimate agents as
\begin{equation}\label{eq:state-leg-error}
	\statelegerr{i}{t+1} \doteq \left|\ls \statecontrleg{t+1} - \one\statelegtruess \rs_i\right| = \left|\ls \Wlegerr{t}\stateleg{0} \rs_i\right|
\end{equation}
where
\begin{multline}\label{eq:state-leg-err-matrix}
	\!\Wlegerr{t} \!\doteq \!\prodW{k}{0}{t} + \sum_{k=0}^t \!\lr\prodW{s}{k+1}{t}\rr\!\lam{k}\\
	- \lr\Wlegtrue\rr^{\!\!\infty}.
\end{multline}
%We aim to bound the probability that the error~\eqref{eq:state-leg-error} is not too large at steady state,
%\ie to find suitable $\epsilon$ and $\delta$ such that
%\begin{equation}\label{eq:leg-err-prob-bound}
%	\pr{\max_{i\in\leg}\limsup_{t\rightarrow\infty}\statelegerr{i}{t} > \epsilon} < \delta.
%\end{equation}
%We have the following result.
\begin{lemma}\label{lem:deviation-leg}
	The deviation from nominal consensus due to legitimate agents' contribution can be bounded as
	\begin{equation}\label{eq:leg-err-prob-bound}
		\pr{\limsup_{t\rightarrow\infty}\statelegerr{i}{t} > \epsilon} < \eta\uleg{\epsilon}, \qquad \forall i\in\leg
	\end{equation}
	where we define
	\begin{equation}\label{eq:leg-bound-prob}
		\uleg{\epsilon} \doteq \dfrac{2}{\epsilon}\lr\e^{s(\gamma)}\lr1-\lr\dfrac{1}{\dmax+1}\rr^{\mean{\tf}}\rr + 1 - \vmin\mean{\ell}\rr
	\end{equation}
	with $\vmin\doteq\min_{i\in\leg}v_i$,
	$\ell\doteq\min\{\ell_1,\ell_2\}$,
	and
	\begin{gather}
		s(\gamma)\doteq-\frac{1}{\gamma}-\frac{\ln(1-ce^{-\gamma})}{\gamma}\cdot \frac{1-ce^{-\gamma}}{ce^{-\gamma}}\label{eq:s}\\
		\begin{multlined}[t]
			\ell_1 \doteq\lr1-c\e^{-\gamma(\maxof{\tf}{1})}\rr^{\frac{1}{1-\e^{-\gamma}}} \lr c\e^{-\gamma(\maxof{(\tf-1)}{0})} \right.\\
			\left. + c\lr\dfrac{1-c\e^{-\gamma}}{\dmax+1}\rr^{\tf-1}\dfrac{1-\e^{-\gamma(\tf-1)}}{1-\e^{-\gamma}}\one[\{\tf>1\}] \rr
		\end{multlined}\label{eq:l1}\\
		\ell_2\doteq 1-\e^{s(\gamma)}.\label{eq:l2}
	\end{gather}
	% \textcolor{red}{Luca, can you please verify that this bound is not trivial, i.e., that there are cases where $\eta\uleg{\epsilon}\leq 1$  for $\epsilon<\eta$?}
\end{lemma}
\begin{proof}
	Let us denote
	\begin{equation}
		\Wlegerr{t} = \Wlegerr[1]{t} + \Wlegerr[2]{t}
	\end{equation}
	where
	\begin{comment}
		\canOmit{\LB{(first attempt) ``steady-state'' error 
				\begin{equation}\label{eq:leg-err-contribution-lambda}
					\Wlegerr[1]{t} \doteq \sum_{k=T(t)-1}^t\lr\prodWss{s}{k+1}{t}\rr\lam{k}
				\end{equation}
				and ``transient'' error
				\begin{multline}\label{eq:leg-err-contribution-misclassification}
					\Wlegerr[2]{t} \doteq \sum_{k=0}^{T(t)-2} \lr\prodWss{s}{T(t)}{t}\prodW{s}{k+1}{T(t)-1}\rr\lam{k} \\
					+\prodW{k}{0}{t} - \lr\Wlegtrue\rr^\infty.
				\end{multline}
			}
			\LB{(second attempt)}
			\begin{equation}\label{eq:leg-err-contribution-lambda}
				\Wlegerr[1]{t} \doteq \prodW{k}{0}{t} - \lr\Wlegtrue\rr^\infty
			\end{equation}
			expresses the mismatch with the nominal protocol and
			\begin{equation}\label{eq:leg-err-contribution-misclassification}
				\Wlegerr[2]{t} \doteq \sum_{k=0}^t \lr\prodW{s}{k+1}{t}\rr\lam{k}
			\end{equation}
			is associated with the input $\lam{t}\stateleg{0}$ that anchors the legitimate agents to their initial condition throughout.}
	\end{comment}
	\begin{equation}\label{eq:leg-err-contribution-lambda}
		\Wlegerr[1]{t} \doteq \prodW{k}{0}{t} - \left(\prod_{k=0}^t(1-\lambda_k) \right)\lr\Wlegtrue\rr^\infty
	\end{equation}
	expresses the mismatch with the nominal (true) weights, and
	\begin{multline}\label{eq:leg-err-contribution-misclassification}
		\Wlegerr[2]{t} \doteq \sum_{k=0}^t \lr\prodW{s}{k+1}{t}\rr\lam{k}\\
		-\lr1-\prodlam{k}{0}{t}\rr\lr\Wlegtrue\rr^\infty
	\end{multline}
	is associated with the input $\lam{t}\stateleg{0}$ that anchors the legitimate agents to their initial condition throughout.
	Similarly to the analysis in~\autoref{sec:convergence},
	convergence to a consensus can be established a.s. for each of the deviation terms respectively associated with $\Wlegerr[1]{t}$ and $\Wlegerr[2]{t}$.
	\begin{comment}
		\LB{(fourth attempt)
			Let's try to bound the whole thing at once.
			We have
			\begin{equation}\label{eq:state-leg-err-matrix}
				\begin{aligned}
					\Wlegerr{t} &\doteq 
					\prodW{k}{0}{t} + \sum_{k=0}^t \lr\prodW{s}{k+1}{t}\rr\lam{k} - \lr\Wlegtrue\rr^\infty\\
					&= \lr\Wlegtrue\rr^{t-\tf} \ls\prodlam{k}{\tf}{t}\lr\prodW{k}{0}{\tf-1} 
					+ \sum_{k=0}^{\tf-2} \lr\prodW{s}{k+1}{\tf-1}\rr\lam{k} + \sum_{k=\tf-1}^t\lam{k}I\rr
					- \lr\Wlegtrue\rr^{\tf}\rs
				\end{aligned}
			\end{equation}
			and the lower bound on the diagonal element is
			\begin{equation}
				(1-c\e^{-\gamma \tf})^{\frac{1}{1-\e^{-\gamma}}} \lr \lr\dfrac{1-c}{\dmax+1}\rr^{\tf}\one[\{\tf>0\}] + c\lr\dfrac{1-c\e^{-\gamma}}{\dmax+1}\rr^{\tf-1}\dfrac{1-\e^{\gamma(\tf-1)}}{1-\e^{-\gamma}}\one[\{\tf>1\}] + \dfrac{c\e^{-\gamma\max\{\tf-1,0\}}}{1-\e^{-\gamma}} + \one[\{\tf=0\}]\rr
			\end{equation}
			Here,
			the same considerations as for $\ell_1$ above apply.
		}
	\end{comment}

\journalVersion{}{%
		Before proceeding further,
		we note that the almost sure consensus discussed in \autoref{sec:convergence} translates into almost sure existence of the limit $\lim_{t\rightarrow\infty}\stateerr{t}{i}$.
		Moreover,
		it can be shown by the same arguments that the quantities $\Wlegerr[1]{t}\stateleg{0}$ and $\Wlegerr[2]{t}\stateleg{0}$ almost surely convergence to vectors with all finite and equal elements.
		This allows us to formally simplify~\eqref{eq:leg-err-markov-limsup} from the limit supremum to the limit.
		Indeed,
		the law of total probability yields
		\begin{equation}\label{eq:limsup-prob-partition}
			\begin{aligned}
				\pr{\limsup_{t\rightarrow\infty}\:\statelegerr{i}{t} > \epsilon}
				&=\pr{\limsup_{t\rightarrow\infty}\:\statelegerr{i}{t} > \epsilon \: \cap \: A}\\
				&+ \pr{\limsup_{t\rightarrow\infty}\:\statelegerr{i}{t} > \epsilon \: \cap \: B}\\
				&+ \pr{\limsup_{t\rightarrow\infty}\:\statelegerr{i}{t} > \epsilon \: \cap \: C},
			\end{aligned}
		\end{equation}
		where the events $A,B,C$ form a partition and are defined as
		\begin{align}
			&A : \exists \lim_{t\rightarrow\infty}\:\ls \Wlegerr{t,1}\stateleg{0} \rs_i 
   \wedge \exists  \lim_{t\rightarrow\infty}\:\ls \Wlegerr{t,2}\stateleg{0} \rs_i \wedge \tf <\infty\\
			&B : \exists \lim_{t\rightarrow\infty}\:\ls \Wlegerr{t,1}\stateleg{0} \rs_i 
   \wedge \exists  \lim_{t\rightarrow\infty}\:\ls \Wlegerr{t,2}\stateleg{0} \rs_i  \wedge \tf =\infty\\
			&C : \nexists \lim_{t\rightarrow\infty}\:\ls \Wlegerr{t,1}\stateleg{0} \rs_i \lor \nexists \lim_{t\rightarrow\infty}\:\ls \Wlegerr{t,2}\stateleg{0} \rs_i .
		\end{align}
		\cref{lem:misclassification-probability} implies that $\pr{\tf<\infty}=1$.
		Moreover,
		\autoref{sec:convergence} shows that $\tf$ finite implies that $\lim_{t\rightarrow\infty}\stateleg{t}$ exists and is a consensus.
		Therefore,
		$\pr{B} = \pr{C} = 0$,
		$\pr{A} = 1$,
		and
		\begin{equation}
			\begin{aligned}
				\pr{\limsup_{t\rightarrow\infty}\:\statelegerr{i}{t} > \epsilon}
				&=\pr{\limsup_{t\rightarrow\infty}\:\statelegerr{i}{t} > \epsilon \: \cap \: A}\\
				&=\pr{\limsup_{t\rightarrow\infty}\:\statelegerr{i}{t} > \epsilon \left|\right. A}\pr{A}\\
    &=\pr{\lim_{t\rightarrow\infty}\statelegerr{i}{t} > \epsilon\left|\right. A}.
%				&\le\pr{\lim_{t\rightarrow\infty}\left| \ls \Wlegerr[1]{t}\stateleg{0}\rs_i \right| + \textcolor{blue}{\lim_{t\rightarrow\infty}}\left| \ls \Wlegerr[2]{t}\stateleg{0}\rs_i \right| > \epsilon} \\
%				&\le\dfrac{1}{\epsilon} \lr\mean{\lim_{t\rightarrow\infty}\left| \ls \Wlegerr[1]{t}\stateleg{0}\rs_i \right|} + \mean{\lim_{t\rightarrow\infty}\left| \ls \Wlegerr[2]{t}\stateleg{0}\rs_i \right|}\rr.
			\end{aligned}
		\end{equation}
		% \textcolor{blue}{A similar} argument used in~\eqref{eq:limsup-prob-partition} can be used \textcolor{blue}{in the analysis of the contribution of the malicious agents}.
		Finally,
		by Markov's inequality we have
		\begin{multline}\label{eq:leg-err-markov-lim}
	\pr{\lim_{t\rightarrow\infty}\statelegerr{i}{t} > \epsilon\left|\right. A}\\
			\le \dfrac{1}{\epsilon} \mean{\lim_{t\rightarrow\infty}\left| \ls \Wlegerr[1]{t}\stateleg{0}\rs_i \right| + \lim_{t\rightarrow\infty}\left| \ls \Wlegerr[2]{t}\stateleg{0}\rs_i \right|\left|\right. A}.
		\end{multline}
	}
For simplicity of notations only, hereafter we omit the conditioning on the event $A$ from terms such as \eqref{eq:leg-err-markov-lim} whenever we utilize the existence of the limit a.s.

	We remark that the following relation can be established as well as a direct consequence of Markov's inequality.
	\begin{align}\label{eq:leg-err-markov-limsup}
		\begin{split}
			\pr{\limsup_{t\rightarrow\infty}\statelegerr{i}{t} > \epsilon} 
			&\le \dfrac{1}{\epsilon} \mean{\limsup_{t\rightarrow\infty}\left| \ls \Wlegerr[1]{t}\stateleg{0}\rs_i \right|} \\
			&+ \dfrac{1}{\epsilon} \mean{\limsup_{t\rightarrow\infty}\left| \ls \Wlegerr[2]{t}\stateleg{0}\rs_i \right|}.
		\end{split}
	\end{align}	
	To retrieve the bound~\eqref{eq:leg-err-prob-bound},
	we evaluate the expected values in~\eqref{eq:leg-err-markov-limsup}.
	To this aim,
	we use the following fact.
	\begin{lemma}[\protect{\!\!\cite[Lemma~4]{Yemini22tro-resilienceConsensusTrust}}]\label{lem:matrix-difference}
		Let $\ell > 0$ and $X,Y\in\Real{N\times N}$ be two sub-stochastic matrices such that $[X]_{ii} \ge \ell$ and $[Y]_{ii} \ge \ell$ for $i = 1,\dots,N$.
		Then,
		it holds $[|X-Y|\one]_{i} \le 2(1-\ell)$ for $i = 1,\dots,N$ where $|A|$ is the matrix with elements $|[A]_{ij}|$.
	\end{lemma}

	\myParagraph{Bound on first term, i.e. $\mean{\limsup_{t\rightarrow\infty}\left|\ls\Wlegerr[1]{t}\stateleg{0}\rs_i\right|}$}
	Let $T(t)$ be the first time instant such that the true weights are recovered through time $t$:
	% \textcolor{red}{Can you elaborate/remark on the difference between $T(t)$ and $\tf$? Also, it seems that $T(t)$ depends on $k$, perhaps we should change the notation to reflect it.}
	\begin{equation}\label{eq:T(t)}
		T(t) \doteq \min \lb k\ge0 : \Wleg{s} = \Wlegtrue, s = k,\dots,t \rb.
	\end{equation}
	If no $k\le t$ achieves the minimum in~\eqref{eq:T(t)},
	we use the convention that $T(t) \doteq t+1$.
	By definition,
	it holds $T(t) \le \tf$ for all $t\ge0$ and $T(t)\equiv\tf$ for $t\ge\tf$ almost surely.
	Define
	\begin{equation}\label{eq:delta-wtilde1}
		\Delta\Wlegerr[1]{t}\doteq \prodlam{k}{0}{t}\left(\prod_{k=0}^{T(t)-1}W_k^{\mathcal{L}}-\prod_{k=0}^{T(t)-1}\overline{W}^{\mathcal{L}}\right).
	\end{equation}
	%	Because there exists a.s. a finite time $\tf\ge0$ such that the true weights are recovered from $\tf$ throughout,
	\journalVersion{}{Hence,
	almost surely $\lim_{t\rightarrow\infty}\Delta\Wlegerr[1]{t}=\Delta\Wlegerr[1]{\infty}$,
	where
	\begin{equation}\label{eq:delta-wtilde1-ss}
		\Delta\Wlegerr[1]{\infty} \doteq\prodlam{k}{0}{\infty}\lr\prod_{k=0}^{\tf-1}W_k^{\mathcal{L}}-\prod_{k=0}^{\tf-1}\overline{W}^{\mathcal{L}}\rr.
	\end{equation}}
	Recall that $\dmax$ denotes the maximal (in-)degree of legitimate agents,
	with $\dmax<N$.
	From \cref{ass:initial-state-bound} and \cref{lem:W-primitive,lem:matrix-difference},
	 it follows a.s. that
	\journalVersion{%
		\begin{equation}\label{eq:delta-w1-ss-bound}
			\limsup_{t\rightarrow\infty}\left|\ls\Wlegerr[1]{t}\stateleg{0}\rs_i\right|\le2\eta\prodlam{k}{0}{\infty}\lr1-\dfrac{1}{\lr\dmax+1\rr^{\tf}}\rr.
	\end{equation}}{%
		\begin{equation}\label{eq:delta-w1-ss-bound}
			\begin{aligned}
				\lim_{t\rightarrow\infty}\left|\ls\Wlegerr[1]{t}\stateleg{0}\rs_i\right| &= \left|\ls \lr\prodWtrue{k}{\tf}{\infty}\rr\Delta\Wlegerr[1]{\infty}\stateleg{0}\rs_i\right|\\
				&\stackrel{(i)}{\le} \max_{i\in\leg} \left|\ls \Delta\Wlegerr[1]{\infty}\stateleg{0}\rs_i\right|\\
				&\stackrel{(ii)}{\le} \eta \max_{i\in\leg}\ls\left|\Delta\Wlegerr[1]{\infty}\right|\one\rs_i\\
				%		&=\eta\prodlam{k}{0}{\infty}\max_{i\in\leg}\ls\left|\prod_{k=0}^{\tf-1}W_k^{\mathcal{L}}-\prod_{k=0}^{\tf-1}\overline{W}^{\mathcal{L}}\right|\one\rs_i\\
				&\stackrel{(iii)}{\le}2\eta\prodlam{k}{0}{\infty}\lr1-\dfrac{1}{\lr\dmax+1\rr^{\tf}}\rr,
			\end{aligned}
		\end{equation}
		where $(i)$ is because $\Wlegtrue$ is stochastic,
		$(ii)$ follows from~\cref{ass:initial-state-bound},
		and $(iii)$ from~\cref{lem:matrix-difference} in view of~\eqref{eq:delta-wtilde1-ss} and the facts (see~\eqref{eq:weights-true} and~\eqref{eq:weights-trust})
		\begin{equation}\label{eq:bound-wii}
			\ls\Wleg{t}\rs_{ii} \ge \dfrac{1}{\dmax+1}, \qquad \ls\Wlegtrue\rs_{ii} \ge \dfrac{1}{\dmax+1}.
	\end{equation}}
	Next,
	we find an upper bound to the infinite product in~\eqref{eq:delta-w1-ss-bound}.
	%The bound~\eqref{eq:delta-w1-ss-bound} is increasing with $\gamma$.
	%We next develop an upper bound that preserves this behavior consistently.
	The following relationship holds:
	\begin{equation}\label{eq:bound-series}
		\begin{aligned}
			\prod_{k=0}^{\infty}\left(1-\lambda_k\right)
			\journalVersion{}{&=\prod_{k=0}^{\infty}\left(1-ce^{-\gamma k}\right)\\
				&= \exp\left(\sum_{k=0}^{\infty}\ln\left(1-ce^{-\gamma k}\right)\right)\\
				&}\leq \exp\left(\int_{k=0}^{\infty}\ln\left(1-ce^{-\gamma (k+1)}\right)\diff k\right).
		\end{aligned}
	\end{equation}
	Define the dilogarithm function $\text{Li}_2(z)\doteq\sum_{k=1}^{\infty}\frac{z^k}{k^2}$,
	then 
	\begin{equation}
		\begin{aligned}
			\int_{k=0}^{\infty}\ln\left(1-ce^{-\gamma (k+1)}\right)\diff k &= -\frac{\text{Li}_2\left(ce^{-\gamma}\right)}{\gamma}\journalVersion{}{\\
				&=-\frac{1}{\gamma}\sum_{k=1}^{\infty}\frac{c^ke^{-\gamma k}}{k^2}}.
		\end{aligned}
	\end{equation}
	Denote \[\overline{s}(x)\doteq\frac{x - x\ln(1 - x) + \ln(1 - x)}{x}.\] 
	By recalling the identity $\sum_{k=1}^{\infty}\frac{x^k}{k(k+1)}=\overline{s}(x)$ for $|x|\leq 1$,
	it follows
	%	we have that $\overline{s}(x)$ is nonnegative and increasing  with $x$ in the literal $[0,1]$. 
	%	Thus, 
	%	$s(\gamma)\doteq-\overline{s}(ce^{-\gamma})$ is nonnegative and increasing with $\gamma$.
	%	%for $\gamma$ in the interval such that $ce^{-\gamma}\in [0,1]$. 
	%	We conclude the derivation of this bound as follows:
	\begin{equation*}\label{eq:inf-prod-upper-bound}
		%		\begin{multlined}
		-\frac{\text{Li}_2\left(ce^{-\gamma}\right)}{\gamma} \leq -\frac{1}{\gamma}\sum_{k=1}^{\infty}\frac{\lr ce^{-\gamma}\rr^k}{k(k+1)}=s(\gamma)\doteq-\frac{\overline{s}(ce^{-\gamma})}{\gamma}.
		%		=-\frac{1}{\gamma}-\frac{\ln(1-ce^{-\gamma})}{\gamma}\cdot \frac{1-ce^{-\gamma}}{ce^{-\gamma}}.
		%		\end{multlined}
	\end{equation*}
	%	\LB{How do you get to the expression of $s(\gamma)$ in~\eqref{eq:s} from $\bar{s}$?
		%		Btw, that's a very cool bound, congrats!}
	%  \textcolor{red}{I played a bit with Mathematica checking the convergence of different upper bounds. Can you please verify it?}	
	Finally,
	from~\eqref{eq:delta-w1-ss-bound}--\eqref{eq:bound-series} and~\eqref{eq:inf-prod-upper-bound},
	the first expectation in~\eqref{eq:leg-err-markov-limsup} can be almost surely upper bounded as follows:
	\journalVersion{
		\begin{equation}\label{eq:w1-ss-bound}
			\mean{\limsup_{t\rightarrow\infty}\left|\ls\Wlegerr[1]{t}\stateleg{0}\rs_i\right|}	\le 2\eta\e^{s(\gamma)} \lr1-\lr\dfrac{1}{\dmax+1}\rr^{\mean{\tf}}\rr.
		\end{equation}
	}{\begin{equation}\label{eq:w1-ss-bound}
			\begin{aligned}
				\mean{\lim_{t\rightarrow\infty}\left|\ls\Wlegerr[1]{t}\stateleg{0}\rs_i\right|}
				&\le 2\eta\e^{s(\gamma)} \lr1-\mean{\dfrac{1}{\lr\dmax+1\rr^{\tf}}}\rr \\%\pr{\tf>0}.\\
				&\le 2\eta\e^{s(\gamma)} \lr1-\lr\dfrac{1}{\dmax+1}\rr^{\mean{\tf}}\rr,
			\end{aligned}
		\end{equation}
		where the second line follows from Jensen's inequality.
	}
	
	\begin{comment}
		\canOmit{\LB{The bound from the third to the fourth line in~\eqref{eq:Y-bound}
				may be tightened.
				Also,
				from the first two lines of~\eqref{eq:Y-bound},
				it also holds
				\begin{equation}
					\begin{aligned}
						[Y]_{ii} 	&= \sum_{k=-1}^{\tf-2} \dfrac{1}{(\dmax+1)^{\tf-k-1}} \lr\prodlam{s}{k+1}{\tf-1}\rr\lam{k}\\
						&> \sum_{k=-1}^{\tf-2} \dfrac{(1-\alpha\e^{-\gamma(k+1)})^{\frac{1-\e^{-\gamma(\tf-k-1)}}{1-\e^{-\gamma}}}}{(\dmax+1)^{\tf-k-1}} \lam{k}.
					\end{aligned}
				\end{equation}
				How can we use this?
			}
			\MY{Are you looking to find $\max_{i\in\leg}\limsup_{t\rightarrow\infty}\mean{\left| \ls \Wlegerr[2]{t}\stateleg{0}\rs_i \right|}$? \\
				If so, can you separate it into two cases 1) $T_f>m$ and $T_f<=m$? where you adjust $m$ to our benefit?}}
	\end{comment}
	
	\myParagraph{Bound on second term, i.e. $\mean{\limsup_{t\rightarrow\infty}\left|\ls\Wlegerr[2]{t}\stateleg{0}\rs_i\right|}$}
	We split the first matrix in $\Wlegerr[2]{t}$ as
	\begin{equation}
		\sum_{k=0}^t \lr\prodW{s}{k+1}{t}\rr\lam{k} = X_{t,1} + X_{t,2}
	\end{equation}
	where
	\begin{subequations}
		\begin{flalign}
			X_{t,1} &\doteq \sum_{k=0}^{T(t)-2} \lr\prodW{s}{k+1}{t}\rr\lam{k} \\
			X_{t,2} &\doteq \sum_{k=\maxof{(T(t)-1)}{0}}^t \lr\prodWss{s}{k+1}{t}\rr\lam{k}.
		\end{flalign}
	\end{subequations}
	\begin{comment}
		\begin{multline}
			\Delta\Wlegerr[2]{t} \doteq \prodlam{k}{T(t)}{t}\sum_{k=0}^t \lr\prodW{s}{k+1}{T(t)-1}\rr\lam{k} \\
			-\lr1-\prodlam{k}{0}{t}\rr\prodWtrue{k}{0}{T(t)-1},
		\end{multline}
		then we have $\Delta\Wlegerr[2]{\infty} \doteq\lim_{t\rightarrow\infty}\Delta\Wlegerr[2]{t}$ with
		\begin{multline}\label{eq:delta-w2-ss}
			\Delta\Wlegerr[2]{\infty}=\prodlam{k}{\tf}{\infty}\sum_{k=0}^\infty \lr\prodW{s}{k+1}{\tf-1}\rr\lam{k}\\
			-\lr1-\prodlam{k}{0}{\infty}\rr\prodWtrue{k}{0}{\tf-1}
		\end{multline}
	\end{comment}
	\journalVersion{}{Similarly to $\Wlegerr[1]{t}$,
	a.s. it holds $\lim_{t\rightarrow\infty}\Wlegerr[2]{t}=\Wlegerr[2]{\infty}$,
	where
	\begin{flalign}
		\Wlegerr[2]{\infty} &= X_{\infty} -\lr1-\prodlam{k}{0}{\infty}\rr\one v^\top\label{eq:w2-ss}\\
		X_\infty&\doteq X_{\infty,1} + X_{\infty,2}\\
		X_{\infty,1}&= \one v^\top\prodlam{k}{\tf}{\infty}\sum_{k=0}^{\tf-2} \lr\prodW{s}{k+1}{\tf-1}\rr\lam{k}\label{eq:X1}\\
		X_{\infty,2} &= \sum_{k=\maxof{(\tf-1)}{0}}^\infty \lr\prodWss{s}{k+1}{\infty}\rr\lam{k}.\label{eq:X2}
	\end{flalign}}
	\journalVersion{Applying~\cref{ass:initial-state-bound} and~\cref{lem:W-primitive,lem:matrix-difference} yields}{It can be verified that $X_\infty$ is sub-stochastic.
		Applying~\cref{ass:initial-state-bound} and~\cref{lem:W-primitive,lem:matrix-difference} yields
	}
	\begin{equation}\label{eq:w2-ss-bound}
		\begin{aligned}
			\journalVersion{\limsup}{\lim}_{t\rightarrow\infty}\left|\ls\Wlegerr[2]{t}\stateleg{0}\rs_i\right| %&= \left|\ls \lr\prodWtrue{k}{\tf}{\infty}\rr\Delta\Wlegerr[2]{\infty}\stateleg{0}\rs_i\right|\\
			\journalVersion{}{&\le \max_{i\in\leg} \left|\ls \Wlegerr[2]{\infty}\stateleg{0}\rs_i\right|\\
				&\le \eta \max_{i\in\leg}\ls\left|\Wlegerr[2]{\infty}\right|\one\rs_i\\
				&}\le \journalVersion{\lim_{t\rightarrow\infty}}{}2\eta (1-\ell\journalVersion{_t}{}),
		\end{aligned}
	\end{equation}
	%where we define $\Wlegerr[2]{\infty} \doteq \lim_{t\rightarrow\infty}\Wlegerr[2]{t}$ with
	where $\ell\journalVersion{_t}{}$ is a lower bound on the diagonal elements of each of the two matrices \journalVersion{in~\eqref{eq:leg-err-contribution-misclassification}.}{in~\eqref{eq:w2-ss}.}
	%\begin{equation}\label{eq:w2-ss-bound}
	%	\mean{\max_{i\in\leg}\limsup_{t\rightarrow\infty}\left|\ls\Wlegerr[2]{t}\stateleg{0}\rs_i\right|} \le 2\eta\lr1-\mean{\ell}\rr
	%\end{equation}
	%\textcolor{red}{Looking at \eqref{eq:leg-err-markov} we don't have to compute the limit. Is this correct? If so, it may be helpful  to calculate the expected value of the nonlimiting form of \eqref{eq:bound} using the law of total expectation where the conditional event is $\tf<h$, assuming that the probability of the event $\tf\geq h$ (which we can upper bound) is sufficiently small.}
	%\LB{I actually didn't write the $\limsup$ and the $\max$ in \eqref{eq:leg-err-markov} for brevity, but I guess the idea is to bound the maximum deviation at steady state, as initially stated in \eqref{eq:leg-err-prob-bound}. I guess the writing is still formally correct, but we can fix it as well. Will the law of total expectation still hold at the limit for $t\rightarrow\infty$?}
	As for the second matrix,
	it holds
	\begin{equation}\label{b2}
		\ls\left(1-\prod_{k=0}^\infty(1-\lambda_k) \right)\one v^\top\rs_{ii} \ge \ell_2\vmin,
	\end{equation}
	where $\vmin$ and $\ell_2$ are defined in~\cref{lem:deviation-leg}.
	%where $\vmin\doteq\min_{i\in\leg}v_i$.
	\journalVersion{We separately bound the diagonal elements of $X_{t,1}$ and $X_{t,2}$.}{We separately bound the diagonal elements of $X_{\infty,1}$ and $X_{\infty,2}$.}
	\journalVersion{It holds}{Consider the inequality
		\begin{equation}\label{eq:inf-prod-lam-inequality}
			\prodlam{s}{k}{K-1} > (1-\lam{k})^{\frac{1}{\lam{k}}\sum_{s=k}^{K-1}\lam{s}} = (1-c\e^{-\gamma k})^{\frac{1-\e^{-\gamma(K-k)}}{1-\e^{-\gamma}}}.
		\end{equation}
		From~\eqref{eq:inf-prod-lam-inequality},
		the infinite product in~\eqref{eq:X1} can be bounded as}
	\begin{equation}\label{eq:prod-lam-inequality}
		\journalVersion{\prodlam{k}{T(t)}{t} > \lr1-c\e^{-\gamma T(t)}\rr^{\frac{1-\e^{-\gamma(t+1-T(t))}}{1-\e^{-\gamma}}}.}{\prodlam{k}{\tf}{\infty} > \lr1-c\e^{-\gamma\tf}\rr^{\frac{1}{1-\e^{-\gamma}}}.}
	\end{equation}
	%\textcolor{red}{It may be better to compute/bound $\prodlam{k}{0}{\infty}$ and $\prod_{k=0}^{\tf-1} W_k^{\mathcal{L}}$ separately.}
	%\textcolor{red}{To bound the term $\prodlam{k}{0}{\infty}$ we can try using the identity
		%	\[\prodlam{k}{0}{\infty}=\exp\left(\sum_{k=0}^{\infty}\ln(1-\lambda_k)\right),\]
		%	and then upper bounding $\sum_{k=0}^{\infty}\ln(1-\lambda_k)$ using an integral. Here is a possible useful link for a somewhat similar integral
		%	\url{https://math.stackexchange.com/questions/60478/evaluating-the-integral-int-0-infty-ln-left1-e-x-right-mathrm}
		%}
	Consider now the \journalVersion{following}{} inequality\journalVersion{:}{in $X_{\infty,1}$. It holds}
	%\begin{multline}\label{eq:delta-w2-first}
	%	\sum_{k=0}^\infty \lr\prodW{s}{k+1}{\tf-1}\rr\lam{k}\\
	%	=\sum_{k=0}^{\tf-2} \lr\prodW{s}{k+1}{\tf-1}\rr\lam{k} + \sum_{k=\tf-1}^\infty\lam{k}I.
	%\end{multline}
	%Then,
	\begin{multline}\label{eq:delta-w2-first-bound}
		\journalVersion{\ls\sum_{k=0}^{T(t)-2} \lr\prodW{s}{k+1}{T(t)-1}\rr\lam{k}\rs_{ii}\\
			\ge c\lr\dfrac{1-c\e^{-\gamma}}{\dmax+1}\rr^{T(t)-1} \dfrac{1 - \e^{-\gamma(T(t)-1)}}{1-\e^{-\gamma}}\mathds{1}_{\{T(t)>1\}}.
		}{\ls\sum_{k=0}^{\tf-2} \lr\prodW{s}{k+1}{\tf-1}\rr\lam{k}\rs_{ii}\\
		\begin{aligned}
			&\stackrel{{(i)}}{\ge} \sum_{k=0}^{\tf-2} \lr\prod_{s=k+1}^{\tf-1}\dfrac{1-\lam{s}}{\dmax+1}\rr\lam{k}\\
			&\stackrel{{(ii)}}{\ge} \sum_{k=0}^{\tf-2} \lr\dfrac{1-\lam{k+1}}{\dmax+1}\rr^{\tf-k-1}\lam{k}\\
			&\stackrel{{(iii)}}{\ge} \lr\dfrac{1-\lam{1}}{\dmax+1}\rr^{\tf-1} \sum_{k=0}^{\tf-2} \lam{k}\\
			&\ge c\lr\dfrac{1-c\e^{-\gamma}}{\dmax+1}\rr^{\tf-1} \dfrac{1 - \e^{-\gamma(\tf-1)}}{1-\e^{-\gamma}}\mathds{1}_{\{\tf>1\}}
		\end{aligned}}
	\end{multline}
	\journalVersion{Then, the diagonal elements of $X_{t,1}$ are bounded as}{where $(i)$ follows from~\eqref{eq:bound-wii},
		and $(ii)$ and $(iii)$ because the arguments of product and summation are increasing with the respective indices.
		Additionally, the diagonal elements of $X_{\infty,1}$ are bounded as
	}
	\journalVersion{%
		\begin{multline}\label{eq:b11}
			%	\ls\one v^\top\prodlam{k}{\tf}{\infty}\sum_{k=0}^{\tf-2} \lr\prodW{s}{k+1}{\tf-1}\rr\lam{k}\rs_{ii} \ge\\
			\ls X_{t,1}\rs_{ii} \ge
			\vmin\lr1-c\e^{-\gamma T(t)}\rr^{\frac{1-\e^{-\gamma(t+1-T(t))}}{1-\e^{-\gamma}}}\\
			\cdot c\lr\dfrac{1-c\e^{-\gamma}}{\dmax+1}\rr^{T(t)-1}\dfrac{1-\e^{-\gamma(T(t)-1)}}{1-\e^{-\gamma}}\one[\{T(t)>1\}].
		\end{multline}
	}{%
	\begin{multline}\label{eq:b11}
		%	\ls\one v^\top\prodlam{k}{\tf}{\infty}\sum_{k=0}^{\tf-2} \lr\prodW{s}{k+1}{\tf-1}\rr\lam{k}\rs_{ii} \ge\\
		\ls X_{\infty,1}\rs_{ii} \ge
		\vmin\lr1-c\e^{-\gamma \tf}\rr^{\frac{1}{1-\e^{-\gamma}}}  c\lr\dfrac{1-c\e^{-\gamma}}{\dmax+1}\rr^{\tf-1}\\
		\cdot\dfrac{1-\e^{-\gamma(\tf\-1)}}{1-\e^{-\gamma}}\one[\{\tf>1\}].
	\end{multline}
	}
	We bound the diagonal elements of \journalVersion{$X_{t,2}$}{$X_{\infty,2}$} as
	\begin{equation}\label{eq:b12}
		\begin{aligned}
			\journalVersion{\ls X_{t,2}\rs_{ii}}{\ls X_{\infty,2}\rs_{ii}&\ge \ls\lr\prodWss{k}{\maxof{\tf}{1}}{\infty}\rr\lam{\maxof{(\tf-1)}{0}}\rs_{ii} \\}
			&\ge\vmin\journalVersion{%
				\begin{multlined}[t]
					\lr1-c\e^{-\gamma(\maxof{T(t)}{1})}\rr^{\frac{1-\e^{-\gamma(t+1-(\maxof{T(t)}{1}))}}{1-\e^{-\gamma}}}\\
					\cdot c\e^{-\gamma(\maxof{(T(t)-1)}{0})}
				\end{multlined}}{\prodlam{k}{\tf}{\infty} > \lr1-c\e^{-\gamma\tf}\rr^{\frac{1}{1-\e^{-\gamma}}}}
		\end{aligned}
	\end{equation}
	and %the overall bound for the diagonal elements of $X_{\infty}$ is 
	\begin{equation}
		\journalVersion{%
			\ls X_{t}\rs_{ii} = \ls X_{t,1}\rs_{ii} + \ls X_{t,2}\rs_{ii} \ge \ell_{t,1}\vmin}{%
			\ls X_{\infty}\rs_{ii} = \ls X_{\infty,1}\rs_{ii} + \ls X_{\infty,2}\rs_{ii} \ge \ell_{1}\vmin
		}
	\end{equation}
	where \journalVersion{$\ell_{t,1}$ is given by~\eqref{eq:b11}--\eqref{eq:b12}.
		Further, $\lim_{t\rightarrow\infty}\ell_{t,1}=\ell_1$ because $\lim_{t\rightarrow\infty}T(t)=\tf<\infty$ almost surely.}{$\ell_1$ is defined in~\eqref{eq:l1}.}
	\journalVersion{It follows $\lim_{t\rightarrow\infty}\ell_t = \vmin\min\{\ell_1,\ell_2\}$ in~\eqref{eq:w2-ss-bound}.}{The two matrices in~\eqref{eq:w2-ss} have diagonal elements lower bounded by $\ell = \min\{\ell_1,\ell_2\}$.}
	Then,
	the second expectation \journalVersion{in~\eqref{eq:leg-err-markov-limsup} can be a.s.}{in~\eqref{eq:leg-err-markov-lim} can be} bounded as
	\begin{equation}\label{eq:w2-ss-bound-exp}
		\mean{\journalVersion{\limsup}{\lim}_{t\rightarrow\infty}\left|\ls\Wlegerr[2]{t}\stateleg{0}\rs_i\right|} \le 2\eta\lr1 - \vmin\mean{\min\{\ell_1,\ell_2\}}\rr.
	\end{equation}
	Finally,
	the probability~\eqref{eq:leg-err-prob-bound} can be bounded almost surely at the limit according to\journalVersion{~\eqref{eq:leg-err-markov-limsup}}{~\eqref{eq:leg-err-markov-lim}}
	by plugging in~\eqref{eq:w1-ss-bound} and~\eqref{eq:w2-ss-bound-exp}.
\end{proof}

A few remarks are in order to understand the meaning of bound~\eqref{eq:leg-bound-prob} and how it behaves as $\gamma$ varies.
For convenience,
we recall the expression of the bound below:
\begin{equation}
	\pr{\limsup_{t\rightarrow\infty}\statelegerr{i}{t} > \epsilon} < \eta\uleg{\epsilon}, \ \uleg{\epsilon} \propto \e^{s(\gamma)} - \mean{\ell}.
\end{equation}
The behavior of the term $\uleg{\epsilon}$ is mainly affected by two functions of $\gamma$,
which are $\e^{s(\gamma)}$ and $\mean{\ell}$.
%as regards dependence on $\gamma$,
%the  function  while the second is proportional to the negative expectation of $\ell$ w.r.t. $\tf$.

The first function,
ruled by $s(\gamma)$,
expresses the deviation due to following the protocol \eqref{eq:update-rule-regular} with the learned weights~\eqref{eq:weights-trust} rather than with the (unknown) true weights~\eqref{eq:weights-true},
and it is increasing with $\gamma$.
In words,
this suggests that setting $\gamma$ small (\ie making $\lam{t}$ decay slowly overtime)
is beneficial to performance because legitimate agents can learn the trustworthy neighbors while keeping balanced weights (thus avoiding biases caused by misclassification of legitimate neighbors) during this learning process.
This is reminiscent of the strategy in~\cite{Yemini22tro-resilienceConsensusTrust},
where the consensus starts at $T_0$ and a larger value of $T_0$ reduces the deviation term associated with data exchange among legitimate agents. %makes the agents later but with a higher confidence in the learned weights.
Moreover,
the coefficient that multiplies $\e^{s(\gamma)}$ increases with $\mean{\tf}$,
so that,
for any choice of $\gamma$,
the deviation is larger for larger $\tf$.
%Note that,
%if the legitimate agents immediately recover the whole network ($\mean{\tf}=0$),
%the first terms trivially vanishes.

The second function appearing in $\uleg{\epsilon}$ is proportional to the negative expectation of $\ell$ w.r.t. $\tf$ and expresses the impact of the input term $\lam{t}\state{0}{i}$ in~\eqref{eq:update-rule-regular} that anchors the legitimate agents to their initial condition.
It is not easy to analytically evaluate the minimum $\ell=\min\{\ell_1,\ell_2\}$,
in general.
Nonetheless,
the following facts hold:
\begin{enumerate}[label=(\arabic*)]
	\item the term $\ell_2$ is strictly decreasing with $\gamma$;
	\item the term $\ell_1$ is strictly increasing with $\gamma$ for $\tf\le1$,
	while for $\tf>1$ it is the product between an increasing and a decreasing function;
	\item the limits of the two terms evaluate
	\begin{subequations}
		\begin{gather}
			\ell_1^0\doteq\lim_{\gamma\rightarrow0}\ell_1 = 0\\
			\ell_2^0\doteq\lim_{\gamma\rightarrow0}\ell_2 = \vmin\\
			\ell_1^\infty\doteq\lim_{\gamma\rightarrow\infty}\ell_1 = \vmin c\one[\{\tf\le1\}]+\dfrac{\vmin c\one[\{\tf>1\}]}{\lr\dmax+1\rr^{\tf-1}}\label{eq:ell1-lim-inf}\\
			\ell_2^\infty\doteq\lim_{\gamma\rightarrow\infty}\ell_2 = c\vmin
		\end{gather}
	\end{subequations}
	and it follows
	\begin{equation}\label{eq:limits-ineq}
		\ell_1^0 < \ell_2^0, \qquad \ell_1^\infty \le \ell_2^\infty
	\end{equation}
	where the equality $\ell_1^\infty = \ell_2^\infty$ holds if and only if $\tf\le1$;
	\item from~\eqref{eq:ell1-lim-inf},
	it follows that $\ell_1^\infty$ decreases with $\tf$ and $\lim_{\tf\rightarrow\infty}\ell_1^\infty=0$.
	%	Moreover,
	%	for $\tf>0$,
	%	we have $\ell_1^\infty > \ell_2^0$.
\end{enumerate}
From the items (1)--(3) above and continuity of $\ell_1$ and $\ell_2$,
we infer the following result,
summarized as a lemma.
\begin{lemma}\label{lem:bound-deviation}
	If $\tf\le1$,
	it holds $\ell=\ell_1$.
	If $\tf>1$,
	there exist $\bar{\gamma}_1>0$ and $\bar{\gamma}_2>0$ such that 
	%	it holds $\ell \rightarrow\ell_1^\infty$ as $\gamma\rightarrow0$ and $\ell \rightarrow \ell_2^\infty$ as $\gamma\rightarrow\infty$.
	%	Further,
	%	by continuity,
	%	\eqref{eq:limits-ineq} implies 
	$\ell=\ell_1$ for $\gamma<\bar{\gamma}_1$ and $\gamma>\bar{\gamma}_2$.
	%	for some with $\bar{\gamma}_1\le\bar{\gamma}_2$.
	%	and the lower bound is $\ell=\ell_1$ for $\gamma<\bar{\gamma}_1$ and $\ell = \ell_2$ for $\gamma>\bar{\gamma}_2$.
\end{lemma}
\cref{lem:bound-deviation} implies that,
for $\tf\le1$,
$\ell=\ell_1$ and thus the term $-\ell$ is decreasing with $\gamma$.

On the other hand,
%while we can state that $\ell=\ell_1$ for small and large values of $\gamma$ in the case $\tf>1$, %and $\gamma>\gamma_2$,
%for suitable values of $\gamma_1>0$ and $\gamma_2>0$,
it is analytically difficult to infer how the term $\ell_1$ (and hence $\ell$) behaves for $\tf>1$ and generic values of $\gamma$.
%From the facts above and continuity of $\ell_1$ and $\ell_2$,
%we can infer the following result.
%\begin{lemma}\label{lem:bound-deviation}
%	For $\tf=0$,
%	it holds $\ell=\ell_1$.
%	For $\tf>0$,
%	we have $\ell \rightarrow\ell_1^\infty$ as $\gamma\rightarrow0$ and $\ell \rightarrow \ell_2^\infty$ as $\gamma\rightarrow\infty$.
%	Further,
%	by continuity,
%	\eqref{eq:limits-ineq} implies $\ell_1<\ell_2$ for $\gamma<\bar{\gamma}_1$ and $\ell_2<\ell_1$ for $\gamma>\bar{\gamma}_2$,
%	for some finite $\bar{\gamma}_1$ and $\bar{\gamma}_2$ with $\bar{\gamma}_1\le\bar{\gamma}_2$,
%	and the lower bound is $\ell=\ell_1$ for $\gamma<\bar{\gamma}_1$ and $\ell = \ell_2$ for $\gamma>\bar{\gamma}_2$.
%\end{lemma}
%An immediate consequence of~\cref{lem:bound-deviation} is that the bound on deviation due to legitimate agents~\eqref{eq:bound} is strictly decreasing with $\gamma$ if $\tf=0$ and it has a nontrivial minimum point for $\tf>0$.
%On the other hand,
Nonetheless,
the fact (4) above suggests that,
as $\tf$ grows,
$\ell_1$ should have a non-monotonic behavior and admit a nontrivial maximum.
Indeed,
numerical tests show that $-\ell_1$ is minimized at a finite value of $\gamma$ for every $\tf>1$,
and that such a minimizer decreases as $\tf$ increases,
as illustrated in~\autoref{fig:ell}.
Considering that the bound~\eqref{eq:leg-bound-prob} is proportional to $-\ell$ in expectation,
this further suggests that the deviation decreases for small values of $\gamma$ and increases for large values,
with the minimum point that shifts towards $\gamma=0$ as the more likely values of $\tf$ increase.
This behavior means that, %indeed matches intuition: %consistent with numerical simulations:
%if %$\tf=0$,
%%it means that 
%the legitimate agents immediately identify their adversaries,
%they should ideally remove $\lam{t}$ to follow the standard consensus with learned weights.
%%so that the deviation steadily decreases with $\gamma$.
%Conversely,
%if $\tf>0$,
if the legitimate agents need time to learn which neighbors are trustworthy,
they should act more cautiously and increase the parameter $\lam{t}$,
especially at the beginning.
%where it exhibits a nontrivial minimum point that shifts leftwards as $\tf$ increases.

\begin{figure}
	\centering
	\includegraphics[width=.6\linewidth]{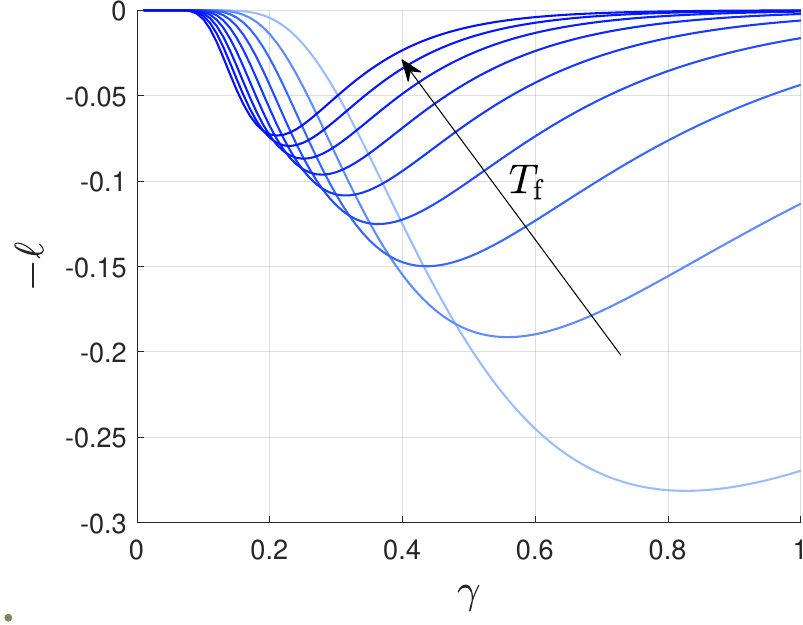}
	\caption{Profile of $-\ell$ in~\eqref{eq:leg-bound-prob} as a function of $\gamma$ with $\tf \in \{2,\dots,10\}$ (the arrow indicates how the curve varies as $\tf$ grows).
		Recall that the bound on (probability of) deviation due to legitimate agents is proportional to $-\mean{\ell}$.
	}
	\label{fig:ell}
\end{figure}

To summarize,
the bound~\eqref{eq:leg-bound-prob} on the deviation due to misclassifications of legitimate agents is numerically seen to be quasi-convex with $\gamma$,
with the minimum point that approaches $\gamma=0$ as $\tf$ increases,
\ie according to how difficult learning the true weights $\Wlegtrue$ is.

\subsubsection{Malicious Agents}\label{sec:deviation-malicious}

Because malicious agents cannot be trusted,
the protocol~\eqref{eq:update-rule-regular} should ideally annihilate their contribution to legitimate agents' states.
Hence,
we define the deviation term due to malicious agents as
\begin{equation}\label{eq:state-mal-error}
	\statemalerr{i}{t+1} \doteq \left|\ls \statecontrmal{t+1}\rs_i\right|. % =
	%	\left|\ls \sum_{k=0}^t \lr\prodW{s}{k+1}{t}\rr (1-\lam{k})\Wmal{k}\statemal{k}\rs_i\right|.
\end{equation}
%Let $\eta$ bound from above the values transmitted by every agent (including malicious ones):
We have the following result.
\begin{lemma}\label{lem:deviation-mal}
	The deviation from nominal consensus due to malicious agents' contribution can be bounded as
	\begin{equation}\label{eq:state-mal-error-bound-prob}
		\pr{\limsup_{t\rightarrow\infty}\statemalerr{i}{t} > \epsilon} \le \eta\umal{\epsilon}, \qquad \forall i\in\leg
	\end{equation}
	where
	\begin{equation}
		\umal{\epsilon} \doteq \dfrac{L\min\{\dmax,M\}}{2\epsilon}\xi
	\end{equation}
	and we define 
	\begin{equation}\label{eq:xi}
		\xi \doteq \dfrac{1}{\e^{2\meanmal^2}-1} 
		- \dfrac{c(1+\e^{-\gamma})}{\e^{2\meanmal^2}-\e^{-\gamma}}
		+ \dfrac{c^2\e^{-\gamma}}{\e^{2\meanmal^2}-\e^{-2\gamma}}.
	\end{equation}
	% \textcolor{red}{Luca, can you please elaborate more on the significance of \eqref{eq:xi} and explain how it can be reduced beyond the trivial case where the RHS of \eqref{eq:state-mal-error-bound-prob} is greater or equal to 1?}
\end{lemma}
\begin{proof}
	From~\eqref{eq:leg-state-contr-mal} and~\eqref{eq:state-mal-error},
	it follows
	\journalVersion{
		\begin{equation}\label{eq:state-mal-error-bound}
			\statemalerr{i}{t+1} \le \eta \sum_{k=0}^t (1-\lam{k+1}) (1-\lam{k})\max_{i\in\leg}\ls \Wmal{k}\one\rs_i.
		\end{equation}
	}{\begin{equation}\label{eq:state-mal-error-bound}
			\begin{aligned}
				\statemalerr{i}{t+1} &= \left|\ls \sum_{k=0}^t \lr\prodW{s}{k+1}{t}\rr (1-\lam{k})\Wmal{k}\statemal{k}\rs_i\right|\\
				&\stackrel{(i)}{\le} \sum_{k=0}^t \left|\ls  \lr\prodW{s}{k+1}{t}\rr (1-\lam{k})\Wmal{k}\statemal{k}\rs_i\right|\\
				&\stackrel{(ii)}{\le} \eta \sum_{k=0}^t \ls \lr\prodW{s}{k+1}{t}\rr (1-\lam{k})\Wmal{k}\one\rs_i\\
				&\stackrel{(iii)}{\le} \eta \sum_{k=0}^t (1-\lam{k+1}) (1-\lam{k})\max_{i\in\leg}\ls \Wmal{k}\one\rs_i
			\end{aligned}
		\end{equation}
		where 
		$(i)$ follows from the triangle inequality,
		$(ii)$ from \cref{ass:initial-state-bound},
		and $(iii)$ because $\{\Wleg{t}\}_{t\ge0}$ are sub-stochastic matrices and $\lam{t}$ is a decreasing sequence with $0 < 1-\lam{t} < 1$.}
	The weights given to malicious agents are bounded as
	\begin{equation}\label{eq:weight-mal-bound}
		\ls \Wmal{t}\one\rs_i = \sum_{j=1}^M \ls \Wmal{t}\rs_{ij} \le \sum_{j\in\mal}\dfrac{1}{2}\one[\{\beta_{ij}(t)\ge0\}].
	\end{equation}
	Further,
	\begin{equation}
		\max_{i\in\leg}\ls \Wmal{k}\one\rs_i \le \sum_{i\in\leg} \ls \Wmal{k}\one\rs_i.
	\end{equation}
	Thus,
	\begin{flalign}
		\statemalerr{i}{t+1}\journalVersion{\leq\phi(t)\doteq}{&\leq 
		\eta\sum_{k=0}^t (1-\lam{k+1}) (1-\lam{k}) \sum_{i\in\leg}\ls \Wmal{k}\one\rs_i\nonumber\\
		&\leq}
		\dfrac{\eta}{2} \sum_{k=0}^{t} (1-\lam{k+1}) (1-\lam{k}) \sum_{\substack{i\in\leg\\j\in\mal}}\one[\{\beta_{ij}(t)\ge0\}]\journalVersion{}{\doteq \phi(t)}.
		% &\leq
		% \dfrac{\eta}{2} \sum_{k=0}^{\infty} (1-\lam{k+1}) (1-\lam{k}) \sum_{\substack{i\in\leg\\j\in\mal}}\one[\{\beta_{ij}(t)\ge0\}].
	\end{flalign}
	%	Recall that $\trusthist{i}{j}{k}<0$ for $k\ge T(t)$ and $j\in\mal$.
	It follows
	\begin{equation}\label{eq:state-mal-error-mean-bound}
		\begin{aligned}
			\mean{\statemalerr{i}{t+1}}	
				% &\le  \mean{\eta\sum_{k=0}^t (1-\lam{k+1}) (1-\lam{k}) \max_{i\in\leg}\ls \Wmal{k}\one\rs_i}\\
				% &\le  \mean{\eta\sum_{k=0}^t (1-\lam{k+1}) (1-\lam{k}) \sum_{i\in\leg}\ls \Wmal{k}\one\rs_i}\\
				&\le \mean{\phi(t)}\journalVersion{}{\\
				&= \dfrac{\eta}{2} \sum_{k=0}^{t} (1-\lam{k+1}) (1-\lam{k}) \sum_{\substack{i\in\leg\\j\in\mal}}\pr{\beta_{ij}(k)\ge0}\\
				&\le \dfrac{\eta}{2} \sum_{k=0}^{t} (1-\lam{k+1}) (1-\lam{k}) \sum_{\substack{i\in\leg\\j\in\mal\cap\neigh{i}}}\!\!\e^{-2(k+1)\meanmal^2}\\
				&}\le \dfrac{\eta L\min\{\dmax,M\}}{2} z_t\\
		\end{aligned}
	\end{equation}
	where we define %the sequence $\{z_t\}_{t\ge0}$ as
	\begin{equation}
		z_t \doteq \sum_{k=0}^{t} (1-\lam{k+1}) (1-\lam{k}) \e^{-2(k+1)\meanmal^2}.
	\end{equation}
	Note that $\phi$ is non-decreasing with $t$.
	Hence,
	by the monotone convergence theorem,
	we can exchange the expectation in~\eqref{eq:state-mal-error-mean-bound} with the limit for $t\rightarrow\infty$.
	This yields
	\begin{equation}\label{eq:bound-exp}
		\begin{aligned}			\mean{\limsup_{t\rightarrow\infty}\statemalerr{i}{t}} 	\journalVersion{\le}{&\le\mean{\limsup_{t\rightarrow\infty}\phi(t)}\\
				&\stackrel{(i)}{=}\mean{\lim_{t\rightarrow\infty}\phi(t)}\\
				&\stackrel{(ii)}{=}\lim_{t\rightarrow\infty}\mean{\phi(t)}\\
				&\le\dfrac{\eta L\min\{\dmax,M\}}{2}\lim_{t\rightarrow\infty}z_t\\
				&=}\dfrac{\eta L\min\{\dmax,M\}}{2}\xi.
		\end{aligned}
	\end{equation}
	where $(i)$ is because $\phi(t)$ admits a limit and $(ii)$ follows from the monotone convergence theorem.
	Bound~\eqref{eq:state-mal-error-bound-prob} follows  applying Markov inequality to~\eqref{eq:bound-exp}.
\end{proof}

The bound $\umal{\epsilon}$ in~\eqref{eq:state-mal-error-bound-prob} %decreasing with $\alpha$ and 
increases with $\gamma$ through the parameter $\xi$.
This is intuitive:
if $\lam{t}$ is larger,
the legitimate agents are less sensitive to their neighbors' states as per~\eqref{eq:update-rule-regular},
thus they are also more resilient against malicious transmissions and the corresponding deviation term $\statemalerr{i}{t}$ is smaller.
Also,
$\umal{\epsilon}$ increases with $\meanmal$, % and $\tf$,
suggesting that higher uncertainty in classification of malicious agents (represented by greater $\meanmal$) yields a larger deviation,
on average.
				%!TEX ROOT = ../resilient_consensus_trust.tex

\subsubsection{Bound on Deviation}\label{sec:deviation-total}

The overall bound on the deviation from nominal consensus can be computed by merging the two bounds obtained for legitimate and malicious agents' contributions.
Applying the triangle inequality to~\eqref{eq:deviation},~\eqref{eq:state-leg-error}, and~\eqref{eq:state-mal-error} yields
\begin{equation}\label{eq:state-err-triangle-ineq}
	\stateerr{i}{t} \le \statelegerr{i}{t} + \statemalerr{i}{t}.
\end{equation}
%Then,
%applying the union bound yields
%\begin{multline}\label{eq:deviation-bound-prob}
%	\pr{\max_{i\in\leg}\limsup_{t\rightarrow\infty}\stateerr{i}{t} > \epsilon} \le
%	\pr{\max_{i\in\leg}\limsup_{t\rightarrow\infty}\statelegerr{i}{t} > \dfrac{\epsilon}{2}}\\
%	 + 
%	\pr{\max_{i\in\leg}\limsup_{t\rightarrow\infty}\statemalerr{i}{t} > \dfrac{\epsilon}{2}}.
%\end{multline}
%where
%\begin{multline}
%	\uleg{\epsilon} \doteq \dfrac{1}{\epsilon} \lr \eta\dfrac{\alpha + \alpha\e^{-\gamma} - \alpha^2\e^{-\gamma}}{1-\e^{-2\gamma}} \right.\\
%	\left. +2\eta \lr 1- \ell(1)\lr1-\pr{\tf>1}\rr\rr \rr
%\end{multline}
%and
%\begin{equation}
%	\umal{\epsilon} \doteq \dfrac{L\min\{\dmax,M\}\eta}{2\epsilon}\xi.
%\end{equation}
We have the following result that quantifies how distant from the nominal consensus the legitimate agents eventually get.
\begin{thm}[Deviation from nominal consensus]\label{thm:bound}
	The deviation from nominal consensus is upper bounded as
	\begin{equation}\label{eq:deviation-bound-prob}
		\pr{\limsup_{t\rightarrow\infty}\stateerr{i}{t} > \epsilon} \le \eta u(\epsilon), \qquad i\in\leg
	\end{equation}
	with
	\begin{equation}\label{eq:deviation-bound}
		u(\epsilon)\doteq\uleg{\dfrac{\epsilon}{2}} + \umal{\dfrac{\epsilon}{2}}.
	\end{equation}
\end{thm}
\begin{proof}
	It follows by applying the union bound to~\eqref{eq:state-err-triangle-ineq} and then invoking~\cref{lem:deviation-leg,lem:deviation-mal}.
\end{proof}
We can assess the impact of a specific choice of $\lam{t}$ by observing the overall deviation bound~\eqref{eq:deviation-bound-prob}.
Recall that,
in light of the expression~\eqref{eq:lambda},
larger values of $\gamma$ correspond to faster decay of $\lam{t}$ --- \ie the standard consensus protocol is recovered more quickly.
In view of what remarked for the two bounds $\uleg{\epsilon}$ and $\umal{\epsilon}$,
the bound $u(\epsilon)$ above suggests that the steady-state deviation from nominal consensus decreases for small values of $\gamma$ and increases as $\gamma$ is chosen larger.
The presence of a nonzero point of minimum,
which intuitively corresponds to an optimal design of $\gamma$,
is caused by the input term $\lam{t}\state{0}{i}$ added to the standard consensus in~\eqref{eq:update-rule-regular} to enhance resilience,
and represents a possible loss in performance due to forcing a suboptimal protocol for too long compared to the time needed for correct detection of adversaries.
In particular,
the term $\ell$ appearing in $\uleg{\epsilon}$ (see~\eqref{eq:leg-bound-prob}) suggests that the optimal $\gamma$ decreases as $\tf$ increases,
reflecting the need of legitimate agents to act more cautiously when the uncertainty in the trust variables $\trust{i}{j}{t}$ is higher.
On the other hand,
the term $\umal{\epsilon}$ requires $\gamma$ to be small (slow decay of $\lam{t}$) to annihilate the effect of malicious agents.

\begin{rem}[Nominal scenario]\label{rem:no-malicious-deviation}
	In the case with no malicious agents ($M=0$),
	the deviation term $\statemalerr{i}{t}$ is identically zero and the choice of $\gamma$ affects the deviation only through the term $\statelegerr{i}{t}$ due to misclassifying legitimate agents.
\end{rem}
	%!TEX ROOT = ../resilient_consensus_trust.tex

\section{Numerical Simulations}\label{sec:simulations}

To test the effectiveness of the proposed resilient consensus protocol and the design insight suggested by the bound proposed in~\cref{thm:bound},
we run numerical simulations with a sparse network with $50$ legitimate agents and $10$ malicious agents.
The communication links are modeled via a random geometric graph with communication radius equal to $0.2$,
the agents being spread across the ball $[0,1] \times [0,1] \in \Real{2}$.
The initial states of legitimate agents are randomly drawn from the uniform distribution $\mathcal{U}(0,\eta)$ with $\eta=1$,
while the malicious agents follow an oscillatory trajectory about the mean value $2\statelegtruess$ (twice the nominal consensus value) under additive zero-mean Gaussian noise with standard deviation $0.05$.
Note that,
in the absence of data-driven detection mechanisms (the malicious agents are classified based on the trust information $\trust{i}{j}{t}$ that comes from physical transmissions and not based on the states they transmit),
this behavior is most harmful because it steadily drives the legitimate agents far away from the nominal consensus value.
Also,
it holds $\statelegtruess\in(0,\nicefrac{1}{2})$ and the random oscillations of malicious agents are small compared to their mean value,
which verifies~\cref{ass:initial-state-bound}.

We run the proposed protocol~\eqref{eq:update-rule-regular} with $\lam{t}$ according~\eqref{eq:lambda} with $c=0.9$
for $T=1000$ iterations and average all results across $1000$ Monte Carlo runs.
We report four different setups with different values of $\meanleg$ and $\meanmal$ that respectively increase from $0.55$ to $0.7$ and decrease from $0.45$ to $0.3$.
In all experiments,
the trust observations of legitimate (resp., malicious) transmissions are drawn from the uniform distribution centered at $\meanleg$ (resp., $\meanmal$) with length equal to twice the minimum between $1-\meanleg$ and $\meanmal$.

\begin{figure}
	\centering
	\includegraphics[width=\linewidth]{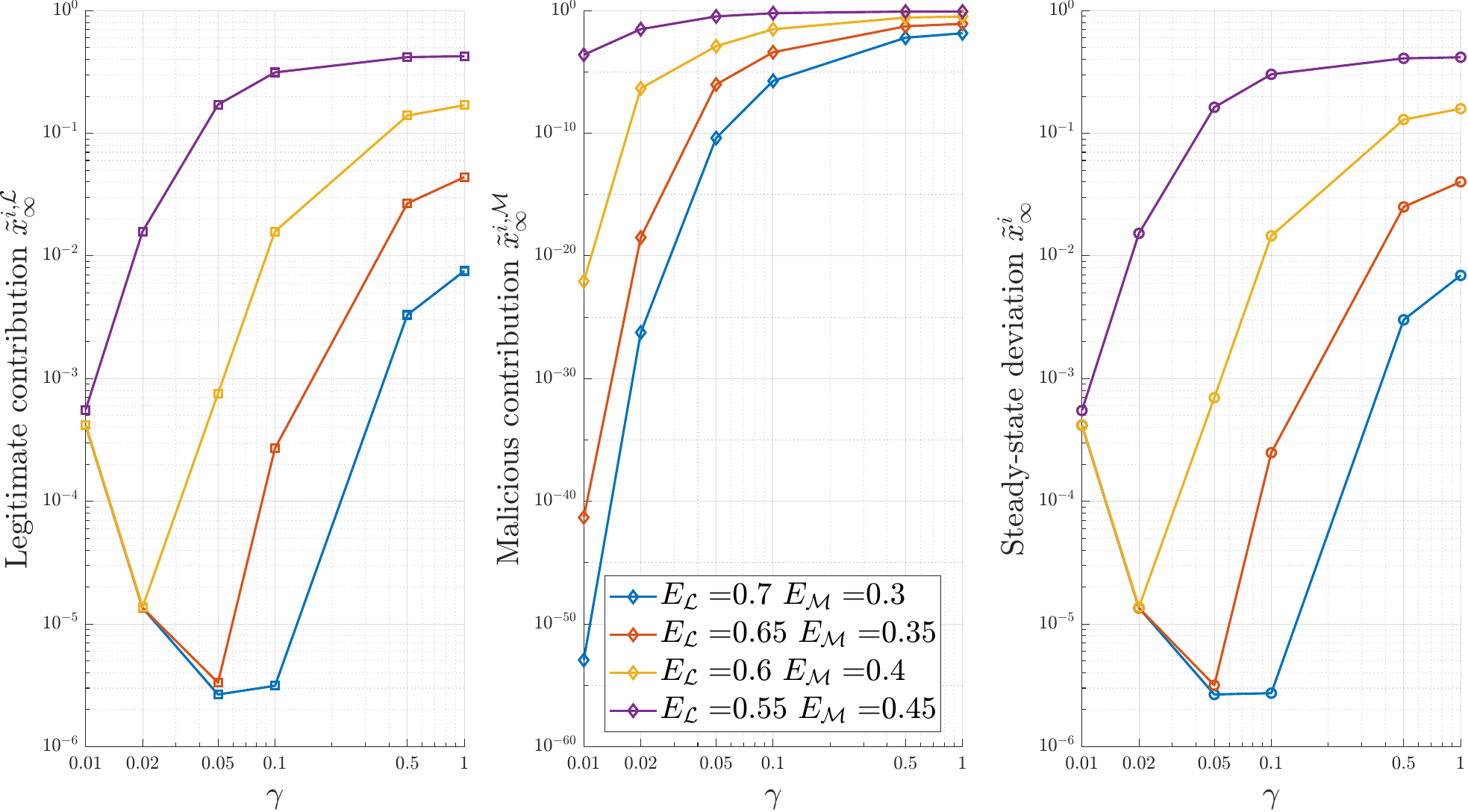}
	\caption{Steady-state deviation from nominal consensus value (right)
		and contributions due to legitimate agents $\statelegerr{i}{\infty}$ (left)
		and to malicious agents $\statemalerr{i}{\infty}$ (middle)
		averaged over 1000 Monte Carlo runs.
		As predicted by the bound~\eqref{eq:leg-bound-prob},
		the deviation term due to misclassification of legitimate agents is minimized by a (small) positive value of $\gamma$ that decreases as the trust scores get more uncertain,
		while the deviation term due to malicious agents steadily increases as $\gamma$ grows,
		according to bound~\eqref{eq:state-mal-error-bound-prob}. 
	}
	\label{fig:deviation}
\end{figure}

The outcomes are depicted in~\autoref{fig:deviation} that shows overall steady-state deviation $\stateerr{i}{\infty}$~\eqref{eq:deviation} in the right box
together with maximal deviation due to legitimate agents $\statelegerr{i}{\infty}$~\eqref{eq:state-leg-error} in the left box
and maximal deviation due to malicious agents $\statemalerr{i}{\infty}$~\eqref{eq:state-mal-error} in the middle box.
It can be seen that the simulated behavior agrees with the analytical bound in~\cref{thm:bound}:
indeed,
the deviation term due to malicious agents steadily increases as $\gamma$ grows,
while the deviation associated with misclassification of legitimate neighbors is minimized by a nonzero value of $\gamma$
that decreases with the uncertainty of trust variables.
For example,
when $\meanleg=0.7$ and $\meanmal=0.3$,
the trust scores are very informative and the deviation is minimized at $\gamma=0.05$,
which dictates a relatively fast decay of the parameter $\lam{t}$.
Conversely,
in the case $\meanleg=0.55$ and $\meanmal=0.45$,
the trust variables are more uncertain and the optimal choice is given by $\gamma=0.01$,
corresponding to a much slower decay of $\lam{t}$.
Moreover,
as the uncertainty in the trust variables increase,
it is more difficult for legitimate agents to correctly classify their neighbors,
which leads to the monotonic increase observed across all deviation terms for every choice of $\gamma$.
	%!TEX ROOT = ../resilient_consensus_trust.tex

\section{Conclusions}\label{sec:conclusions}

We propose a resilient consensus protocol that uses trustworthiness information derived from the physical transmission channel to progressively detect malicious agents,
and complements this information with a time-varying scaling that accounts for how confident the agent is about its neighbors being malicious or not.
Analytical results demonstrate that the proposed protocol leads to a consensus almost surely.
Also,
the asymptotic deviation is upper bounded by a non-monotonic function of the decay rate of the confidence parameter.
Numerical results corroborate these findings,
suggesting that the confidence parameter can be optimally tuned so as to minimize the steady-state deviation.
	
	%!TEX root = ../resilient_consensus_trust.tex

\section*{Acknowledgments}

We thank Prof. Stephanie Gil for the fruitful discussions about incorporating confidence into the trust-based resilient consensus protocol of~\cite{Yemini22tro-resilienceConsensusTrust} and on design considerations for the confidence parameter $\lam{t}$. M.~Yemini additionally thanks Prof. Reuven Cohen for an enriching discussion regarding the finer points of the convergence of random variables and conditional expectations. 
	
	\bibliographystyle{IEEEtran}
	 \journalVersion{
	 	\bibliography{IEEEabrv,
	 		bibfile}
	 	}{
	 	% Generated by IEEEtran.bst, version: 1.14 (2015/08/26)

	 }
		
\end{document}